\documentclass[11pt,letterpaper]{article}
\usepackage{authblk}

\usepackage{thm-restate}

\usepackage[margin=1in]{geometry}
\usepackage[T1]{fontenc}
\usepackage{amsmath,amsthm,mathtools}
\usepackage{newtxtext,newtxmath}
\usepackage{microtype,booktabs,array,tabularx,enumitem}
\usepackage{algorithm,algpseudocode}
\usepackage[numbers,sort&compress]{natbib}
\usepackage{xcolor,hyperref,cleveref,fancyhdr}
\hypersetup{
  hidelinks,
  hypertexnames=false,
  pdftitle={Polynomial Algorithms for Simultaneous
    Share-Based and Envy-Based Guarantees},
  pdfauthor={Hannaneh Akrami, Arash Ashuri, Kurt Mehlhorn}
}
\setlist{itemsep=2pt,topsep=4pt}
\numberwithin{equation}{section}
\newtheorem{theorem}{Theorem}[section]
\newtheorem{observation}[theorem]{Observation}
\newtheorem{lemma}[theorem]{Lemma}
\newtheorem{corollary}[theorem]{Corollary}
\newtheorem{proposition}[theorem]{Proposition}
\theoremstyle{definition}\newtheorem{definition}[theorem]{Definition}
\theoremstyle{remark}\newtheorem{remark}[theorem]{Remark}

\crefname{lemma}{lemma}{lemmas}\Crefname{lemma}{Lemma}{Lemmas}
\crefname{corollary}{corollary}{corollaries}\Crefname{corollary}{Corollary}{Corollaries}
\crefname{proposition}{proposition}{propositions}\Crefname{proposition}{Proposition}{Propositions}
\crefname{definition}{definition}{definitions}\Crefname{definition}{Definition}{Definitions}
\DeclareMathOperator{\MMS}{MMS}
\DeclareMathOperator{\MXS}{MXS}
\DeclareMathOperator{\RMMS}{RMMS}

\DeclareMathOperator{\LGLB}{\ensuremath{LGLB}}

\newcommand{\poly}{\operatorname{poly}}

\newcommand{\Repair}{\mathsf{EFLRepair}}

\newcommand{\Complete}{\mathsf{CompleteEFL}}

\renewcommand{\xi}{\MXS}

\definecolor{darkgreen}{RGB}{0,100,0}

\newcommand{\vefx}{s^{\text{EFX}}}
\newcommand{\vefone}{s^{\text{EF1}}}
\newcommand{\vefl}{s^{\text{EFL}}}

\newcommand{\efxenvy}{EFX-envy}
\newcommand{\efxenvies}{EFX-envies}
\newcommand{\eflenvy}{EFL-envy}
\newcommand{\eflenvies}{EFL-envies}

\newcommand{\A}{\mathbf{A}}
\newcommand{\C}{\mathbf{C}}

\newcommand{\effx}{\textsf{EF2X}}

\newcommand{\efkx}{\textsf{EFkX}}

\newcommand{\APS}{\operatorname{APS}}

\newcommand{\Epoch}{\mathsf{Epoch}}

\numberwithin{equation}{section}
\title{\vspace{-1.2em}\textbf{Polynomial Algorithms for Simultaneous Share-Based and Envy-Based Guarantees}\\[0.35em]
}

\author[1,2]{Hannaneh Akrami \thanks{hakrami@mpi-inf.mpg.de}}
\author[1]{Arash Ashuri\thanks{aashuri@mpi-inf.mpg.de}}
\author[1]{Kurt Mehlhorn \thanks{mehlhorn@mpi-inf.mpg.de}}

\affil[1]{Max Planck Institute for Informatics \& Saarland Informatics Campus, Saarbrücken, Germany}
\affil[2]{Hertz  Chair for Algorithms and Optimization, University of Bonn, Germany}

\date{}
\begin{document}
\maketitle

\begin{abstract}
We study the simultaneous achievement of share-based and envy-based fairness for allocating indivisible goods with nonnegative additive valuations. The AnyPrice share (APS) dominates the maximin share (MMS). The strongest previously known approximation of APS is $\frac{2n}{3n-1}$, achievable in polynomial time~\cite{BEF21}. However, this result does not provide an accompanying envy-based guarantee, and, to the best of our knowledge, no nontrivial approximation of APS was previously known to be achievable simultaneously with envy-freeness up to one less-preferred good (EFL). In a separate line of work, allocations satisfying minimum EFX share (MXS) together with EFL, $\frac{2}{3}$-MMS along with EFL, and later the stronger RMMS together with EFL, were shown to exist. The known constructions for these simultaneous guarantees, however, did not provide even pseudopolynomial running-time bounds.

We strengthen both lines of results. First, for every  $\varepsilon\in(0,1)$, we give an algorithm that computes an EFL and $(1-\varepsilon)$-MXS allocation while guaranteeing
$\frac{2n}{3n-1}$
of every agent's APS. This matches the APS-only factor $\frac{2n}{3n-1}$. The running time is polynomial in the binary input length and $1/\varepsilon$. Thus, this gives the first constant APS guarantee that is simultaneously compatible with exact EFL, while also providing an arbitrarily close approximation to MXS.
We show that our approach for approximate APS cannot be directly generalized to compute an MXS+EFL allocation in polynomial time.

Second, we provide an algorithm that computes an allocation satisfying exact MXS, exact EFL, and
$
\frac{2n}{3n-1}\text{-APS}
$
in pseudo-polynomial time. Hence, we match the strongest previously known APS guarantee while simultaneously imposing exact MXS and exact EFL.

Third, we give an algorithm for exact RMMS together with exact EFL running in pseudopolynomial time for every fixed number of agents.

Finally, we show that it is weakly NP hard to compute a number between MXS and RMMS share given an additive valuation function.
\end{abstract}

    \section{Introduction}\label{sec:intro}

We study complete fair allocations of a set $M$ of $m$ indivisible goods to a set $N$ of $n$ agents, where each agent $i$ has a nonnegative additive valuation function
$v_i:2^M\to\mathbb{R}_{\geq 0}.$
A complete allocation is denoted by
$\mathbf{A}=\langle A_1,\ldots,A_n\rangle,$
where the bundles $A_1,\ldots,A_n$ form a partition of $M$.

A fair allocation should address two different questions. First, does an agent receive sufficiently high value relative to the resources available? Second, can the agent be convinced that another agent's bundle is not substantially better than her own? Share-based and envy-based notions formalize these two questions, respectively. Neither perspective subsumes the other when goods are indivisible. An allocation may give every agent a substantial share of the total resources while still creating significant envy; conversely, an allocation with little envy may give an agent only a small fraction of a natural share benchmark.

\paragraph{Envy-based fairness.}
With indivisible goods, an envy-free allocation need not exist. This has motivated relaxations that allow an agent to disregard one good when comparing her bundle with another agent's bundle. EF1 requires that, for every pair of agents, envy can be eliminated by removing some good from the envied bundle. EFL strengthens EF1 by additionally requiring that the removed good itself is worth no more than the observer's own bundle when there are at least two goods in the envied bundle. EFX is stronger still: envy must disappear after the removal of \emph{any} good from the envied bundle \cite{CKMPSW19}. Thus, $\mathrm{EFX}\ \Longrightarrow\ \mathrm{EFL}\ \Longrightarrow\ \mathrm{EF1}.$ EF1 and EFL admit efficient standalone algorithms~\cite{LMMS04,BBMN18}.

\paragraph{Share-based fairness.}
We next introduce the share benchmarks used in this paper. Let $\Pi_n(M)$ denote the set of all $n$-partitions
$
\mathbf{P}=(P_1,\ldots,P_n)
$
of $M$.

\begin{definition}[Maximin share]
The \emph{maximin share} (MMS) of agent $i$ is $\mu_i = \MMS_n(v_i;M) = \max_{\mathbf{P}\in\Pi_n(M)} \min_{j\in[n]}v_i(P_j).$ Equivalently, agent $i$ partitions the goods into $n$ bundles and then receives a least-valued bundle; $\mu_i$ is the largest value she can guarantee herself by choosing the partition.

For $\alpha\in[0,1]$, an allocation $\mathbf{A}$ is \emph{$\alpha$-MMS} if
$
v_i(A_i)\ge \alpha\mu_i$ for every $i\in N.$
The case $\alpha=1$ is called an MMS allocation.
\end{definition}

MMS fairness was introduced in this context by Budish~\cite{B10}. Exact MMS allocations need not exist, even for additive valuations~\cite{KPW18}, which motivates both approximations of MMS and alternative share notions.

\begin{definition}[AnyPrice share]
The \emph{AnyPrice share} (APS) of agent $i$ with equal entitlement $1/n$ is
\[
\tau_i
=
\APS_{1/n}(v_i;M)
=
\min_{\substack{p_g\ge0\\ \sum_{g\in M}p_g=1}}
\;
\max_{\substack{S\subseteq M\\ \sum_{g\in S}p_g\le1/n}}
v_i(S).
\]
Here the goods are assigned arbitrary nonnegative prices of total price one, and agent $i$ receives a budget of $1/n$. For a fixed price vector, she chooses her most valuable affordable bundle; APS is the value she can guarantee against the worst possible choice of prices~\cite{BEF21}.

For $\alpha\in[0,1]$, an allocation $\mathbf{A}$ is \emph{$\alpha$-APS} if $v_i(A_i)\ge\alpha\tau_i \qquad \text{for every }i\in N.$ \end{definition}

APS is at least as demanding a benchmark as MMS. In fact, $\mu_i\le\tau_i \qquad \text{for every agent }i.$ To see this, fix any normalized price vector and an MMS partition $\mathbf{P}=(P_1,\ldots,P_n)$. 
Since the prices of all goods sum to one, at least one of the $n$ bundles $P_j$ has price at most $1/n$. That bundle is affordable and, because $\mathbf{P}$ is an MMS partition, has value at least $\mu_i$. Hence, under \emph{every} price vector, agent $i$ can afford a bundle worth at least $\mu_i$. Taking the minimum over price vectors gives $\tau_i\ge\mu_i$. Consequently, every $\alpha$-APS allocation is also $\alpha$-MMS.

\begin{definition}[Minimum EFX share]
For a bundle $B$, define $\vefx_i(B) = v_i(B)-\min_{g\in B}v_i(g)$ when $B\neq\varnothing$, and let $\vefx_i(\varnothing)=0$.
Given an $n$-partition $\mathbf{B}=(B_1,\ldots,B_n)$, a bundle $B_a$ is \emph{EFX-feasible for agent $i$} if $v_i(B_a) \ge \max_{j\in[n]}\vefx_i(B_j).$ The \emph{minimum EFX share} (MXS) of agent $i$ is
\[
\xi_i
=
\MXS_n(v_i;M)
=
\min
\left\{
v_i(B_a):
\begin{array}{l}
\mathbf{B}=(B_1,\ldots,B_n)\in\Pi_n(M),\\
B_a\text{ is EFX-feasible for }i
\end{array}
\right\}.
\]
Thus $\xi_i$ is the minimum value of an EFX-feasible bundle over all partitions in which such a bundle is considered~\cite{CGRSV22}.

For $\alpha\in[0,1]$, an allocation $\mathbf{A}$ is \emph{$\alpha$-MXS} if $v_i(A_i)\ge\alpha\xi_i \qquad \text{for every }i\in N.$ The case $\alpha=1$ is called an MXS allocation.
\end{definition}

MXS therefore combines a share-based viewpoint with an EFX certificate: the numerical guarantee depends only on the bundle received by the agent, but the benchmark is defined through the existence of a partition in which that bundle is EFX-feasible. Unlike MMS, MXS is always feasible~\cite{CGRSV22,HANR25}.

A more recent share notion is the \emph{residual maximin share} (RMMS)~\cite{F25}. Its purpose is to capture what happens during a sequential allocation process, when some bundles have already been given to other agents and removed from consideration. To make this precise, fix an agent $i$ and a threshold $t$. Imagine that $k<n$ pairwise disjoint bundles
$D_1,\ldots,D_k$
have already been removed, and suppose that each of them is worth strictly less than $t$ to agent $i$.
The threshold $t$ is called residually feasible if, \emph{for every} $k<n$ and every such collection of removed bundles, the remaining goods
can still be partitioned into $n-k$ bundles, each worth at least $t$ to agent $i$. The RMMS value, which we denote by $\rho_i$, is the largest threshold having this property.

Thus RMMS strengthens the usual maximin reasoning by requiring the same threshold to survive every possible sequence of removals of bundles that the agent herself regards as insufficient. Notice that taking $k=0$ recovers the ordinary partition requirement: if $t$ is residually feasible, then $M$ itself admits an $n$-partition whose every bundle has value at least $t$. Consequently, $\rho_i\le\mu_i$.

On the other hand, Feige~\cite{F25} showed that RMMS dominates MXS,
$\xi_i\le\rho_i$,
and, for additive valuations, that
$\rho_i\ge\frac23\mu_i.$

Putting these share benchmarks together gives the useful hierarchy
\[
\boxed{
\xi_i
\;\le\;
\rho_i
\;\le\;
\mu_i
\;\le\;
\tau_i
}
\qquad
\text{for every agent }i.
\]
The notions nevertheless serve different purposes. MXS is derived from EFX-feasible partitions and is designed to be feasible; RMMS captures the robustness of a share under the removal of low-valued bundles; MMS is the classical divide-and-choose benchmark; and APS is the strongest of these four numerical benchmarks, interpreting fairness through adversarial prices. In particular, an $\alpha$-APS guarantee automatically implies an $\alpha$-MMS guarantee, while an exact RMMS guarantee implies both exact MXS and, for additive valuations, at least $2/3$-MMS.

    \emph{Combining share-based and envy-based fairness.}
    The questions most closely related to this paper concern simultaneous guarantees. Ashuri and Gkatzelis~\cite{AG25} showed that MXS can be satisfied simultaneously with EFL for a valuation class containing additive valuations. Akrami and Rathi~\cite{AR25} established the existence of partial allocations satisfying both $2/3$-MMS and EFX, as well as complete allocations satisfying both $2/3$-MMS and EF1. Feige~\cite{F25} introduced the residual maximin share and showed that it can be achieved simultaneously with EFL, which further were showed that the same algorithm also satisfies EEFX\footnote{A bundle is EEFX-feasible for some agent if there exists a partition including the bundle such that the bundle is EFX-feasible for the agent. An allocation is EEFX if every agent gets an EEFX-feasible bundle.} ~\cite{AFMMR26} for additive valuation functions. 
    None of these result provide an even pseudo-polynomial time algorithm.
    Our results complement this line of work by focusing on the computational complexity of obtaining simultaneous guarantees.
    Beyond additive valuations, Feige and Fine~\cite{FF26} study simultaneous approximate MMS and EF1 guarantees for submodular valuations.

\subsection{Main Results}
Write $\tau_i=\APS_{1/n}(v_i;M)$, $\xi_i=\MXS_n(v_i;M)$, and $\rho_i=\RMMS_n(v_i;M)$. Let $\ell$ be the total binary encoding length. For integer valuations, set $U_i=v_i(M),\qquad U=\max_i U_i.$ The distinction between $U$ and its encoding length is essential. Define $\beta_n=\frac{2n}{3n-1}>\frac23.$ \begin{restatable}{theorem}{mainthone}
\label{theorem:2/3aps bounded}
[Polynomially computable residual APS floors]
    For a nonnegative rational additive valuation $v$, we have 
    $\beta_n\APS_{1/n}(v;M)
    \le\RMMS_n(v;M)$, and
    a threshold $\widehat t(v,n)$ satisfying $\beta_n\APS_{1/n}(v;M) \le\widehat t(v,n) \le\RMMS_n(v;M) \le\MMS_n(v;M)$ can be computed in $O((m+1)^2)$. Moreover, $\LGLB$, which is an algorithm introduced in \cref{sec:certificate} and runs in polynomial time, is a residual splitter at this threshold: after deleting any $k<n$ disjoint bundles, each worth strictly less than $\widehat t(v,n)$, it returns $n-k$ residual bundles each worth at least $\widehat t(v,n)$.
\end{restatable}

    $\widehat t(v,n)$ is an upper bound for $\beta_n\APS_{1/n}(v;M)$. We call it residual APS floor since in our algorithm, we have a threshold $t_i$ for every agent that gradually decrease, 
    and we try to give each agent $i$ value $t_i$, 
    and we make sure that $t_i$ never gets lower than $\widehat t(v_i,n)$.

An $n$-partition of $M$, together with a distinguished bundle $B$ that is EFX-feasible for agent $i$ in that partition, is an \emph{MXS upper witness} for agent $i$. 
By definition of MXS,  if in an allocation $\A$ we have $v_i(A_i)\geq v_i(B)$, then
such a witness certifies that $\xi_i\le v_i(A_i)$. This is particularly useful because computing the MXS value $\xi_i$ itself is weakly NP-hard~\cite{CGRSV22}: an MXS upper witness allows the agent to certify MXS fairness without explicitly computing $\xi_i$.

\begin{restatable}{theorem}{mainthtwo}
\label{thm:main-poly}
[Polynomial time with near-exact MXS]
For nonnegative rational additive valuations and rational $0<\varepsilon<1$, 
there exists an algorithm that computes a complete EFL allocation satisfying
\[
 v_i(A_i)\ge\beta_n\tau_i=\frac{2n}{3n-1}\tau_i \quad \text{and}
 \qquad v_i(A_i)\ge(1-\varepsilon)\xi_i
 \qquad(i\in[n]).
\]
Its running time is polynomial in $n,m,\ell,1/\varepsilon$. It returns for every agent $i$ an \emph{MXS upper witness}
of value $b_i\ge\xi_i$, together with the guarantee $v_i(A_i)\ge(1-\epsilon)b_i$.
The APS guarantee has no multiplicative loss depending on $\varepsilon$. 
\end{restatable}

\begin{restatable}{theorem}{mainththree}
\label{thm:main-exact}
[Exact MXS in uniform pseudopolynomial time]
For nonnegative integer additive valuations, there exists an algorithm that returns a complete EFL allocation satisfying $v_i(A_i)\ge\max\{\xi_i,\beta_n\tau_i\} \qquad(i\in[n])$ in $(U+1)\poly(n,m,\ell)$ time. It returns an MXS upper witness for every agent $i$ with a distinguished bundle $B_i$ such that $v_i(A_i) \geq v_i(B_i)$. 
In particular, it computes $\beta_n$-APS, exact MXS, and exact EFL simultaneously for arbitrary $n$ in pseudopolynomial time.
\end{restatable}

\begin{restatable}{theorem}{mainthfour}
\label{thm:main-rmms}
[Exact RMMS for a fixed number of agents]
For nonnegative integer additive valuations, there exists an algorithm that returns a complete EFL allocation satisfying $v_i(A_i)\ge\rho_i\qquad(i\in[n])$ in $\poly(n,m,U,\ell)(U+1)^{n-1}$ time. This is pseudo-polynomial for every fixed $n$.
\end{restatable}

\begin{table}[t]
\centering\small
\begin{tabularx}{\linewidth}{@{}>{\raggedright\arraybackslash}X>{\raggedright\arraybackslash}X@{}}
\toprule
Guarantees in the same complete allocation & Running time and qualification\\
\midrule
$\beta_n$-APS, $(1-\varepsilon)$-MXS, EFL & $\poly(n,m,\ell,1/\varepsilon)$; rational input\\[4pt]
$\beta_n$-APS, exact MXS, EFL & $(U+1)\poly(n,m,\ell)$; integer input and arbitrary $n$\\[4pt]
Exact RMMS, EFL & $\poly(n,m,U,\ell)(U+1)^{n-1}$; pseudopolynomial for fixed $n$\\
\bottomrule
\end{tabularx}
\caption{All envy guarantees are exact. The first algorithm relaxes only MXS; the third has an explicit value-dependent exponent involving the number of agents.}\label{tab:results}
\end{table}

APS dominates MMS, so any approximate for APS in our statements imply the same MMS factors. 

    Computing minimum EFX share \cite{CGRSV22}, and RMMS value \cite{AFMMR26} are weakly NP-hard.
    Next, we provide a hardness result for computing any value between these values. This shows that our approach cannot 
    be generalized to compute an MXS+EFL algorithm in polynomial time.
    
\begin{restatable}{theorem}{mxsrmmsenhard}
[Hardness of the MXS--RMMS certificate problem]
\label{thm:mxsh:main}
    For every fixed $n\geq3$, it is weakly NP-hard, given a binary-encoded positive
    integer additive valuation $v$ on $M$, to output a rational number $s$
    satisfying
\begin{equation}
  \MXS_n(v;M)\leq s\leq\RMMS_n(v;M).
  \label{eq:mxsh:sandwich}
\end{equation}
The hardness holds even under the promise
$\MXS_n(v;M)=\RMMS_n(v;M)$.

Hardness also holds for nondegenerate positive integer valuations, where
distinct subsets have distinct values. This latter statement does not
require the promise that the two shares coincide.
\end{restatable}

\subsection{Technical overview}
\label{sec:technical-overview}

The existence of complete allocations satisfying both RMMS and EFL is
already known: Feige~\cite[Corollary~8]{F25} proves this for monotone
valuations, and hence also for additive valuations. An algorithmic
implementation of that proof encounters two major computational tasks:
\begin{enumerate}
    \item Computing each agent's RMMS value, which determines whether
    that agent regards a bundle as sufficient.
    \item After removing $k<n$ bundles that an unassigned agent values
    below her RMMS, partitioning the remaining goods into $n-k$ bundles,
    each worth at least her RMMS. An exact MMS-optimal partition of the
    remaining goods---one that maximizes the minimum bundle value---is
    sufficient, although the proof only requires a partition meeting
    the RMMS threshold.
\end{enumerate}
We first prove that, for every nonnegative additive valuation, $\beta_n\APS_{1/n}(v;M)\le\RMMS_n(v;M), \qquad \beta_n=\frac{2n}{3n-1}.$ Consequently, the \emph{existence} of a complete allocation satisfying
$\beta_n$-APS and EFL follows immediately from Feige's result.
However, both RMMS evaluation and finding an RMMS-feasible partition
are weakly NP-hard for additive valuations~\cite[Proposition~21]{F25}.
Thus the two tasks above cannot, in general, be performed in polynomial
time unless $\mathrm{P}=\mathrm{NP}$. Our algorithmic challenge is to
obtain the APS--EFL guarantee without solving either task.

\paragraph{A computable threshold with a polynomial-time splitter.}
Rather than compute APS or RMMS, we compute a threshold
$\widehat t(v,n)$ satisfying $\beta_n\APS_{1/n}(v;M) \le \widehat t(v,n) \le \RMMS_n(v;M).$ Crucially, we also provide a polynomial-time \emph{residual splitter}
for this threshold: after any $k<n$ disjoint bundles of value strictly
below $\widehat t(v,n)$ have been removed, it partitions the remaining
goods into $n-k$ bundles, each worth at least $\widehat t(v,n)$.
Thus we replace both the exact share computation and the residual
maximin optimization.

The splitter uses $\LGLB$, which processes goods in nonincreasing value
order and assigns each good to a currently least-valued bundle.
The key is to identify thresholds for which this simple procedure
remains valid after the required deletions. In \Cref{sec:certificate},
we give two numerical tests using only singleton values, their counts,
and their sums. Every threshold passing these tests has $\LGLB$ as a
residual splitter. Using fractional APS reference families, we prove
that $\beta_n\APS_{1/n}(v;M)$ passes the tests. We then compute the
largest passing threshold by an exact polynomial-time event sweep,
without evaluating APS.

Note that computing the APS value is weakly NP-hard, even for
additive valuations and two agents~\cite{BEF21}, so it is crucial to define this value.

\paragraph{From the splitter to a polynomial-time APS--EFL algorithm.}
For each agent $i$, set the fixed threshold
$t_i^*=\widehat t(v_i,n)$. Our allocation procedure maintains that
assigned agents meet their thresholds and are EFL toward one another,
while every unassigned agent values every assigned bundle below her
own threshold. Therefore, when an unassigned agent divides the free
goods, the residual splitter always succeeds. Its bundles are trimmed
into minimal desired candidates; an EFL-preserving repair improves an
assigned agent, or a rejection-preserving matching assigns at least
one additional agent.

Efficient splitting alone does not bound the number of repairs.
We therefore use an EFL-specific repair and prove that each improvement
crosses one of only $O(m^4)$ fixed utility levels per agent
(Lemma~\ref{lem:breakpoints}). This gives polynomially many rounds. After all agents meet their thresholds, the
remaining goods can be allocated while preserving EFL and every
agent's achieved value (Lemma~\ref{lem:completion}). Since
$t_i^*\ge\beta_n\tau_i$, the resulting complete allocation satisfies
$\beta_n$-APS and EFL in polynomial time.

\paragraph{Adding an approximate MXS guarantee.}
The same splitter also provides MXS certificates. For arbitrary disjoint
bundles $D_1,\ldots,D_k$, let $L$ be the minimum value in the residual
$\LGLB$ partition and define $b=\max\{L,v(D_1),\ldots,v(D_k)\}.$ Our joint certificate theorem (\Cref{thm:joint-certificate}) shows that
\begin{equation}\label{eq:intro-certificate}
    \max\{\MXS_n(v;M),\widehat t(v,n)\}\le b,
    \qquad
    \beta_n\APS_{1/n}(v;M)\le\widehat t(v,n).
\end{equation}
In particular, $\LGLB$ is also a residual splitter at the MXS threshold:
if every deleted bundle is worth less than $\MXS_n(v;M)$, then
$L\ge\MXS_n(v;M)$. Moreover, the deleted bundles together with the residual partition
contain an explicitly identifiable EFX-feasible bundle of value $b$.
This witnesses the upper bound on MXS without computing MXS itself.

One might instead try to compute a fixed threshold between MXS and
RMMS and use the preceding approach without approximation. However,
we prove that finding any number in this interval is weakly NP-hard,
already for three agents (\Cref{thm:mxsh:main}). We avoid this obstacle
by maintaining an upper certificate $b_i\ge\max\{\xi_i,t_i^*\}$ that
is \emph{not} required to lie below RMMS. For a requested accuracy
$\varepsilon\in(0,1)$, set $t_i=\max\{t_i^*,(1-\epsilon)b_i\}.$ We run the allocation procedure at these thresholds. If a residual
division fails for agent $i$, the joint certificate gives a replacement $\max\{\xi_i,t_i^*\}\le b_i'<(1-\epsilon)b_i,$ and we restart with this smaller certificate. The strict contraction
follows because failure is impossible at the fixed threshold $t_i^*$.

Initially, $b_i$ is the minimum value in the full $\LGLB$ partition,
so $\beta_n\tau_i\le b_i\le\tau_i$. Hence a positive-APS agent can
cause only $O(1/\eta)$ contractions; a zero-APS agent starts with
$b_i=0$ and causes none. Each run between restarts takes polynomial
time, so the total running time is polynomial in the input length
and $1/\varepsilon$. The final allocation is complete and EFL, and
satisfies
\[
    v_i(A_i)\ge\beta_n\tau_i,
    \qquad
    v_i(A_i)\ge(1-\epsilon)b_i\ge(1-\varepsilon)\xi_i
    \qquad(i\in N).
\]
The approximation affects only the MXS guarantee: the fixed APS
threshold $t_i^*$ is never scaled down. For integer valuations,
omitting the multiplicative slack instead gives exact MXS in
pseudopolynomial time by counting strict integer decreases of the
certificates.

    \paragraph{Organization.}
    \Cref{sec:certificate} computes the fixed residual APS floors and proves the strengthened joint MXS--APS certificate. \Cref{sec:engine} gives the polynomial-work EFL allocation engine. \Cref{sec:aps-algorithms} combines these ingredients for the two APS--MXS algorithms and the single-epoch APS--EFL result. \Cref{sec:rmms} develops the RMMS reduction and the exact load dynamic program. Finally, \cref{app:mxsh:hardness} shows why our approach cannot be generalized for computing 
    an MXS+EFL allocation in polynomial time by proving that it is weakly NP hard to compute a number between MXS and RMMS share given a valuation function.

    \subsection{Further related work.}
    Our work lies at the intersection of envy-based and share-based fairness. We first review the literature on exact and approximate EFX, maximin-share guarantees,  Any Price Share, and then discuss fair division of chores.
    
    \emph{Exact EFX for goods.}
    The existence of exact EFX allocations depends strongly on the valuation domain. Recently, Akrami et al.~\cite{AMMSW26} showed that EFX allocations need not exist even for three agents with monotone valuation functions. Thus, while EFX fails for unrestricted monotone valuations, its existence for an arbitrary number of agents with additive valuations remains a major open problem. A number of positive results are known for more structured domains. EFX allocations exist for identical monotone valuations~\cite{PR20}, binary valuations~\cite{BSY23}, instances with at most two distinct monotone valuation functions~\cite{M24}, and additive instances with at most three distinct valuation types~\cite{PGNV25}. For three agents, EFX allocations exist under additive valuations~\cite{CGM24}; this result extends to nice cancelable valuations~\cite{BCFF21} and, more generally, to settings in which one agent has an MMS-feasible valuation while the remaining agents have monotone valuations~\cite{ACGMM25}.
    
    Another line of work studies graphical valuation domains. Christodoulou et al.~\cite{CFKS23} established EFX existence for graphical valuations, where each good is relevant to at most two agents and any pair of agents shares at most one relevant good. 
    Subsequent work extended these results to multigraph valuations \cite{ADKMR25,SS25}. Afshinmehr et al.~\cite{AAPM26} proved EFX existence for triangle-free multigraphs under monotone valuations and gave a pseudopolynomial-time algorithm. More recently, Afshinmehr et al.~\cite{AAMMS26} established EFX existence for general multigraph instances with cancelable valuations and provided a polynomial-time algorithm. 
    
    \emph{Approximate EFX.}
    Beyond exact EFX, a substantial literature studies multiplicative approximations. Plaut and Roughgarden~\cite{PR20} established a $1/2$-EFX guarantee for subadditive valuations, while Amanatidis et al.~\cite{AMN20} proved the existence of $1/\phi$-EFX allocations for additive valuations, where $\phi=(1+\sqrt{5})/2$. Stronger guarantees are possible under additional restrictions. Amanatidis et al.~\cite{ARS24} proved that $2/3$-EFX allocations exist for up to seven agents, as well as for several other restricted additive domains. For instances with at most four distinct additive valuation types, the same $2/3$ guarantee holds for an arbitrary number of agents~\cite{HMN25}. Filos-Ratsikas et al.~\cite{FKW26} subsequently extended the $2/3$ guarantee to eight agents.
    
    A related relaxation is envy-freeness up to any $k$ goods, denoted $\efkx$. The special case $k=2$ is $\effx$. Akrami et al.~\cite{ARS22} proved that $\effx$ allocations exist for any number of agents with restricted additive valuations, while Kaviani et al.~\cite{KSS24} established existence for $(\infty,1)$-bounded valuations. Ashuri et al.~\cite{AGS25} proved that $\effx$ allocations exist for four agents with cancelable valuations; for three agents, they obtain a polynomial-time algorithm, in contrast to the known pseudopolynomial algorithms for computing exact EFX allocations in the three-agent setting. In addition, Filos-Ratsikas et al.~\cite{FKW26} study approximate $\efkx$ and show that for every $k\ge2$, a $\frac{k+1}{k+2}$-approximate $\efkx$ allocation exists for additive valuation functions.

    \emph{Maximin-share guarantees.}
    For additive valuations, Kurokawa et al.~\cite{KPW18} established a $2/3$-MMS guarantee, followed by a sequence of improved approximation algorithms~\cite{AMNS17,GHS18,BK20,GMT19,GT20,AGST23}. Akrami and Garg~\cite{AG24} were  
    the first to improve \(3/4\) by a constant independent of the number of agents, obtaining $3/4+3/3836$, Heidari et al.~\cite{HKSS26} subsequently improved the general guarantee to $10/13$, and finally, Huang and Zhou \cite{HZ26} achieved $7/9$. On the impossibility side, the classical $39/40$ upper bound~\cite{FST21} was recently strengthened by Ezra and Garbuz~\cite{EG26}: for four agents, no guarantee above $20/21$ is possible.
    Gupta and Seetharaman~\cite{GS26} further study the complexity of deciding whether an exact MMS allocation exists.

    \emph{Any-Price Share.}
    The AnyPrice share (APS) was introduced by Babaioff, Ezra, and Feige~\cite{BEF21}, who proved its fractional characterization and gave a polynomial-time $\beta_n$-APS guarantee for additive valuations with equal entitlements. Ben Uziahu and Feige~\cite{BUF23} obtained a polynomial-time $1/3$-APS guarantee for submodular valuations with arbitrary entitlements. Kulkarni, Kulkarni, and Mehta~\cite{KKM23} studied APS for binary-marginal submodular and XOS valuations, while Feige and Huang~\cite{FH23} considered APS guarantees for chores. Feige and Grinberg~\cite{FG25} established APS guarantees for subadditive and XOS valuations, including $1/6$ for XOS valuations with arbitrary entitlements and $4/17$ with equal entitlements. Chen, Li, Luo, and Zhang~\cite{CLLZ26} obtained a polynomial-time $1/2$-APS guarantee for binary XOS valuations with arbitrary entitlements. Babaioff and Feige~\cite{BF25} further studied APS within a broader framework of feasible share notions under arbitrary entitlements. Our focus is complementary: we study simultaneous APS guarantees together with EFL and MXS.

    \emph{Fair allocation of chores.}
    A parallel literature studies indivisible chores, where agents have cost rather than utility functions. Zhou and Wu~\cite{ZW24} gave the first nontrivial approximation guarantees for EFX under additive costs, including an $O(n^2)$ approximation for a general number of agents. Garg, Murhekar, and Qin~\cite{GMQ25} substantially improved this to a constant-factor guarantee, proving that $4$-EFX allocations always exist for additive chores. For three agents, Christoforidis and Santorinaios~\cite{CS24} and Afshinmehr et al.~\cite{AADM24} independently studied stronger approximations; the latter prove that a $2$-EFX allocation exists and can be computed in polynomial time even for subadditive cost functions. Garg and Murhekar~\cite{GM26} subsequently extended the $2$-EFX existence guarantee to an arbitrary number of agents.
    
    Exact EFX behaves differently for chores than for goods. Christoforidis and Santorinaios~\cite{CS24} showed that EFX allocations may fail to exist for monotone costs, and Christoforidis~\cite{C26} strengthened this impossibility to subadditive and submodular costs. More recently, He and Tao~\cite{he2026} proved that for every $n\ge4$, there are tri-valued additive instances admitting no EFX allocation. Further nonexistence results were obtained for three-agent submodular instances~\cite{HKPV26} and for binary XOS and binary supermodular costs~\cite{LLTZ26}. For three additive agents, the general existence question remains open, although EFX is known to exist with at most eight chores~\cite{Zhang26}. Positive results are also known for restricted additive costs: Lin et al.~\cite{LWS26} and Akrami et al.~\cite{AACMS26} obtain strong simultaneous EFX, MMS, and efficiency guarantees, including polynomial-time approximation schemes for the MMS component.

\paragraph{AI acknowledgment.}
    The authors used ChatGPT Pro (OpenAI) as a research and writing assistant during the preparation of this paper. In particular, it was used to help explore and formalize the authors' ideas for extending the RMMS--EFL existence proof to an algorithmic, polynomial-time framework. All mathematical arguments, statements, and proofs in the paper were subsequently checked and verified, and polished by the authors, who take full responsibility for the content.
        
\section{Preliminaries}\label{sec:prelim}

    An instance of \emph{discrete fair division} is given by a tuple
    $\langle N, M, \{v_i\}_{i \in N} \rangle$, where
    $N = [n] = \{1,2,\dots,n\}$ denotes the set of agents,
    $M = [m]$ is a finite set of $m$ indivisible goods, and
    $\{v_i\}_{i \in N}$ is a profile of valuation functions.
    Each agent $i \in N$ has an associated additive valuation function
    $v_i : 2^M \rightarrow \mathbb{R}_{\ge 0}$. 

    A partial allocation
    $\mathbf{A}=\langle A_1,\ldots,A_n\rangle$
    consists of a set $H\subseteq N$ of \emph{assigned agents}, where each agent $i\in H$ is assigned a bundle $A_i$, and  bundles are pairwise disjoint. Agents in $N\setminus H$ are unassigned. The set of unallocated goods is the free pool
    $$
    F=M\setminus\bigcup_{i\in H} A_i.
    $$
    The allocation is \emph{complete} if $H=N$ and $F=\varnothing$.
    Given an allocation and $i\in H$, agent $i$'s value is $v_i(A_i)$.

    An indexed $k$-partition
    $\mathbf{X}=(X_1,\ldots,X_k)$ of a subset of goods $S$ is 
    an ordered set of disjoint subsets of goods such that $\cup_{i\in [k]} X_i = S$; empty occurrences are allowed and retain their indices.
    We write $S+g=S\cup{g}$ and $S-g=S\setminus{g}$.

    \paragraph{Fairness scores.}
    For a nonempty bundle $B$, and valuation function $v$  
    $$
     \vefx_v(B)=v(B)-\min_{g\in B}v(g),\qquad
     \vefone_v(B)=v(B)-\max_{g\in B}v(g),
    $$
    and set both scores to zero for the empty bundle, and we write $\vefx_i(B), \vefone_i(B)$ for $v=v_i$.
    For $|B|\ge2$, put    
    $$
     \vefl_i(B)=\min_{g\in B}\max\{v_i(g),v_i(B-g)\};
    $$
    set $\vefl_i(B)=0$ when $|B|\le1$, and write
    $\vefl_i(\mathbf{X})=\max_j\vefl_i(X_j)$.
    The allocation
    $\mathbf{A}=\langle A_1,\ldots,A_n\rangle$
    is EFX, EF1, or EFL when each assigned agent's value dominates, respectively, every other bundle's $\vefx_i$, $\vefone_i$, or $\vefl_i$ score.

    Given two bundles $S,T\subseteq M$ and an agent $i$, we say that agent $i$ \emph{envies} bundle $T$ with respect to bundle $S$ if
    $v_i(S)<v_i(T)$.
    We say that agent $i$ \efxenvies\ bundle $T$ with respect to bundle $S$ if
    $v_i(S)<\vefx_i(T)$. 
    Equivalently, there exists a good $g\in T$ such that
    $v_i(S)<v_i(T-g)$. Similarly, agent $i$ \eflenvies\ bundle $T$ with respect to bundle $S$ if
    $v_i(S)<\vefl_i(T)$.
    
    Given a partition $\mathbf{X}=(X_1,\ldots,X_k)$, a bundle $X_a$ is \emph{EFX-feasible} for agent $i$ in $\mathbf{X}$ if agent $i$ does not \efxenvy\ any bundle $X_j$ with respect to $X_a$; equivalently,
    
    $$
    v_i(X_a)\ge \max_{j\in[k]}\vefx_i(X_j).
    $$
    
    Likewise, $X_a$ is \emph{EFL-feasible} for agent $i$ in $\mathbf{X}$ if agent $i$ does not \eflenvy\ any bundle $X_j$ with respect to $X_a$; equivalently,
    
    $$
    v_i(X_a)\ge \max_{j\in[k]}\vefl_i(X_j).
    $$

    Since EFX implies EFL and EFL implies EF1, every EFX-feasible bundle is also EFL-feasible. Each of the scores $\vefx_i$, $\vefone_i$, and $\vefl_i$ is nondecreasing under the addition of goods.

    \paragraph{Share based notions.}
    For an integer $k\ge 1$, let $\Pi_k(M)$ denote the set of all indexed $k$-partitions of $M$:
    
    $$
    \Pi_k(M)=
    \{
    \mathbf{W}=(W_1,\ldots,W_k):
    W_a\cap W_b=\varnothing \text{ for }a\neq b,\ 
    \textstyle\bigcup_{a=1}^k W_a=M
    \}.
    $$
    
    The $k$-maximin share of agent $i$ is
    $$
    \MMS_k(v_i;M)
    =
    \max_{\mathbf{W}\in\Pi_k(M)}
    \min_{a\in[k]} v_i(W_a).
    $$
    
    That is, agent $i$ partitions the goods into $k$ bundles and then receives a least-valued bundle according to her valuation. We write  
    $$
    \mu_i=\MMS_n(v_i;M)  \label{eq:mms-def}
    $$
    for the maximin share of agent $i$ in the original instance.

    The \emph{AnyPrice share} (APS) of agent $i$ is defined through normalized nonnegative prices on the goods. Consider any price vector $(p_g)_{g\in M}$ satisfying $p_g\ge0 \qquad\text{and}\qquad \sum_{g\in M}p_g=1.$ Under these prices, agent $i$ has budget $1/n$ and may choose any bundle whose total price is at most $1/n$. Her AnyPrice share is the minimum, over all such price vectors, of the maximum value she can obtain with this budget:
    \begin{align}
     \tau_i=\APS_{1/n}(v_i;M)=\min_{\substack{p_g\ge0\\\sum_g p_g=1}}
                \max_{\substack{S\subseteq M\\\sum_{g\in S}p_g\le1/n}}v_i(S).
     \label{eq:aps-def}
    \end{align}

\begin{definition}[Fractional APS reference family]
Fix an agent $i$ and a threshold $\lambda>0$. A \emph{fractional APS reference family at threshold $\lambda$} is a finite collection $\mathcal{R} = \{(B_j,\gamma_j)\}_{j=1}^r,$ where each $B_j\subseteq M$ is a bundle and each $\gamma_j>0$ is its \emph{reference coefficient}, such that $\sum_{j=1}^r\gamma_j=n,$ \[
\text{$g$'s incident} :=\sum_{j:\,g\in B_j}\gamma_j\le1
\qquad\text{for every }g\in M,
\]
and $v_i(B_j)\ge\lambda \qquad\text{for every }j\in[r].$ We call the bundles $B_j$ the \emph{reference bundles}. For any subfamily indexed by $J\subseteq[r]$, its \emph{reference mass} is $\sum_{j\in J}\gamma_j.$ \end{definition}

\citet{BEF21} showed that a fractional APS reference family at threshold $\lambda$ exists if and only if $\lambda\le\tau_i$. 
We provide its proof in \cref{subsection:fractional aps certificates} for completeness.

An MMS partition is a special, integral APS reference family. Indeed, let $\mathbf{W}=(W_1,\ldots,W_n)$ be an MMS partition for agent $i$, so that $v_i(W_j)\ge\mu_i \qquad\text{for every }j\in[n].$ Set $B_j=W_j \qquad\text{and}\qquad \gamma_j=1 \qquad(j\in[n]).$ Then the total reference mass is $n$, every good belongs to exactly one reference bundle and therefore has total incidence $1$, and every reference bundle has value at least $\mu_i$. Hence an MMS partition is a fractional APS reference family at threshold $\mu_i$.

In this sense, fractional APS reference families generalize MMS partitions: an MMS partition consists of $n$ disjoint reference bundles, each with coefficient one, whereas APS reference bundles may overlap, provided that the total coefficient with which any good is used is at most one. This proves $\mu_i\le\tau_i$.

    The minimum EFX share is
    \begin{equation}\label{eq:mxs-def}
 \xi_i
 =\min\left\{v_i(B_a):
 B\in\Pi_n(M),\quad v_i(B_a)\ge\max_j\vefx_i(B_j) \right\}.
\end{equation}
    Thus, $\xi_i$ is the minimum value, over all $n$-partitions of $M$, of a bundle that is EFX-feasible for agent $i$ in that partition. An allocation
    $\mathbf{A}=\langle A_1,\ldots,A_n\rangle$
    satisfies MXS if
    
    $$
    v_i(A_i)\ge \xi_i
    \qquad\text{for every }i\in N.
    $$
    An $n$-partition of $M$, together with a distinguished bundle $B$ that is EFX-feasible for agent $i$ in that partition, is an \emph{MXS upper witness} for agent $i$ in an allocation $\A$ if $v_i(A_i)\geq v_i(B)$.

\begin{definition}[Residual maximin share~\cite{AFMMR26}]
\label{def:rmms}
    A threshold $t\ge0$ is \emph{residual self-feasible} for $(v,M,n)$ if, for every $0\le k<n$ and every collection of disjoint bundles $D_1,\ldots,D_k$ with $v(D_j)<t$, the remaining set admits an $(n-k)$-partition all of whose bundles have value at least $t$. The residual maximin share $\RMMS_n(v;M)$ is the largest such threshold.
\end{definition}

    We use $c_j=\frac{3j-1}{2j}=\frac32-\frac1{2j},\qquad \beta_j=1/c_j \quad(j\ge1).$ The sequence $c_j$ is strictly increasing. For $j\ge2$, $1<c_j<3/2$.

\section{A joint MXS--APS certificate from residual LGLB}\label{sec:certificate}
\label{app:apsx}

For a valuation $v$, a set of goods $F$, and an integer $q\ge1$, the \emph{largest-good-to-least-bundle} (LGLB) routine
$
\LGLB(v,F,q)
$
starts with $q$ empty bundles. It considers the goods in $F$ in nonincreasing order of their value under $v$ and, one by one, assigns each good to a currently minimum-value bundle, breaking ties arbitrarily. The resulting bundles form a $q$-partition of $F$.

The purpose of this section is to establish the certificates that will be used by both of our APS algorithms. First, for every valuation $v$, we compute a fixed threshold $\widehat t(v,n)$ satisfying $\beta_n\APS_{1/n}(v;M) \le \widehat t(v,n) \le \RMMS_n(v;M) \le \MMS_n(v;M).$ Moreover, after removing any $k<n$ disjoint bundles whose values are strictly below $\widehat t(v,n)$, LGLB partitions the remaining goods into $n-k$ bundles, each of value at least $\widehat t(v,n)$. Thus the threshold comes with a polynomial-time residual splitter, without computing APS or RMMS.

Second, suppose that arbitrary disjoint bundles $D_1,\ldots,D_k$  have been removed, and let $L$ be the minimum value in the residual LGLB partition. We show that the quantity $b=\max\{L,v(D_1),\ldots,v(D_k)\}$ simultaneously gives an upper bound on MXS and on the fixed APS threshold $\widehat t(v,n)$. The same construction explicitly produces an MXS upper witness of value exactly $b$. 

The proof has three ingredients. First, \Cref{lem:aps-fractional} recalls the fractional characterization of APS. Since the reference bundles may overlap, \Cref{lem:fractional-count} gives the fractional singleton-counting argument that we need. Second, we use these references to show that the target $\beta_n\APS_{1/n}(v;M)$ satisfies two numerical tests. We prove that every threshold passing these tests has the required residual LGLB splitter, and compute the largest passing threshold by an exact event sweep. Finally, \Cref{lem:LGLB-score} supplies the explicit MXS witness, and the residual-splitting property gives the strengthened joint certificate in \Cref{thm:joint-certificate}.

Throughout this section, fractional families are used only in the analysis of APS. All LGLB partitions and all bundles produced by the algorithms consist of the original indivisible goods. For $n=1$, take $\widehat t(v,1)=v(M)$; all claims are immediate. Henceforth $n\ge2$. Recall that $c_j=\frac{3j-1}{2j},\qquad \beta_j=\frac1{c_j},$ and that $c_j$ is increasing in $j$.

\subsection{Fractional APS certificates}\label{subsection:fractional aps certificates}
    Next, we present the fractional characterization of APS, 
    and we defer its proof to \cref{sec:aps certificate appendix}.
\begin{restatable}{lemma}{characlemma}
[Fractional characterization~\cite{BEF21}]\label[lemma]{lem:aps-fractional} 
    For a positive threshold $\lambda$, 
    \begin{enumerate}
        \item  $\lambda\le\APS_{1/n}(v;M)$ 
    if and only if there exists a fractional APS reference family at threshold $\lambda$. Equivalently, there are bundles $B_1,\ldots,B_r$ and positive coefficients $\gamma_1,\ldots,\gamma_r$ such that
            \begin{equation}\label{eq:fractional}
            \sum_{j=1}^r\gamma_j=n,\qquad
            \sum_{j:g\in B_j}\gamma_j\le1\quad(g\in M),\qquad
            v(B_j)\ge\lambda\quad(j\in[r]).
            \end{equation}
        \item Furthermore, for any subfamily indexed by $J\subseteq[r]$ whose coefficients sum to $p$, $v\!\left(\bigcup_{j\in J}B_j\right) \ge\sum_{j\in J}\gamma_jv(B_j) \ge p\lambda.$ \item Also, for any fixed set of $h$ goods, the total reference mass of bundles that contain at least one of those goods is at most $h$.       
    \end{enumerate}
\end{restatable}
    
    An immediate consequence is that a positive APS requires at least $n$ positive-valued goods. Indeed, every positive-valued bundle $B_j$ contains at least one positive-valued good, while each good can support coefficient mass at most one; supporting total coefficient mass $n$ therefore requires at least $n$ positive-valued goods.
    
\subsection{Counting singletons in a fractional family}

For an integral partition, when there are fewer than two goods per bundle on average, many bundles must be singletons.
For a fractional APS reference, bundles may overlap, so this argument cannot be applied to the number of bundles directly. We instead count coefficient mass, which gives the fractional analogue in the next lemma.

\begin{lemma}[Fractional singleton counting]\label[lemma]{lem:fractional-count}
    Let $p\le \ell<2p$ be integers with $p\ge1$. Let $G=[\ell]$ consist of $\ell$ distinguished goods with values $1 \geq v_1 \geq v_2 \geq \ldots \geq v_\ell \geq 0$, and fix $\lambda>1$. Suppose that each subset $R\subseteq G$ is assigned a nonnegative coefficient $\gamma_R\ge0$, and that these coefficients satisfy
    \begin{equation}\label{eq:count-input}
     \sum_R\gamma_R=p,\qquad
     \sum_{R:g\in R}\gamma_R\le 1  \text{ (for any } g\in G\text{)}.
    \end{equation}
    Then, 
    \begin{equation}\label{eq:fractional-count}
     \sum_R\gamma_R\bigl(\lambda-v(R)\bigr)_+ + \sum_{j=1}^{2p-\ell}v_j \ge (2p-\ell)\lambda,
    \end{equation}
    where $\bigl(\lambda-v(R)\bigr)_+ = \max(0, \lambda-v(R))$.
\end{lemma}

\begin{proof}
For a bundle $R\subseteq G$, its contribution toward the target $\lambda$ is $\min\{\lambda,v(R)\}$. 
Define the total capped contribution of the fractional family by $C(\gamma) = \sum_R\gamma_R\min\{\lambda,v(R)\}.$ For every $R$, $\min\{\lambda,v(R)\} = \lambda-\bigl(\lambda-v(R)\bigr)_+.$ Since the total coefficient mass is $p$, 
\begin{align}\label{ineq}
    C(\gamma) = p\lambda - \sum_R\gamma_R\bigl(\lambda-v(R)\bigr)_+.
\end{align}
We now derive an upper bound on $C(\gamma)$.

\paragraph{Step 1: saturating every good's incidence.}
Fix a good $g\in G$, and suppose its current incidence is $z=\sum_{R:g\in R}\gamma_R<1.$ The total coefficient mass of bundles not containing $g$ is $p-z$, and since
$p\ge1$,
we get $p-z\ge1-z.$
Hence we can choose exactly $1-z$ units of coefficient mass from bundles not
containing $g$.

If the required amount cuts through a coefficient $\gamma_R$, say only
$\alpha<\gamma_R$ is needed, we \emph{split} this coefficient into two copies, $(R,\gamma_R) \quad\longrightarrow\quad (R,\alpha),\ (R,\gamma_R-\alpha).$ We then add $g$ only to the first copy:
\[
(R,\alpha),\ (R,\gamma_R-\alpha)
\quad\longrightarrow\quad
(R\cup\{g\},\alpha),\ (R,\gamma_R-\alpha).
\]
This increases the incidence of $g$ by exactly $\alpha$. For every other good
$h\neq g$, the two copies still contribute a total coefficient $\alpha+(\gamma_R-\alpha)=\gamma_R,$ so the incidence of $h$ is unchanged. Doing this until a total mass
$1-z$ is selected raises the incidence of $g$ exactly to $1$, preserves the incidence of every other good, and preserves the total coefficient mass. Moreover, adding a nonnegative-weight good cannot decrease $\min\{\lambda,v(R)\}$. 
Hence $C(\gamma)$ does not decrease.

Repeating this operation for every good, we may assume that every good is saturated: $\sum_{R:g\in R}\gamma_R=1 \qquad\text{for every }g\in G.$ \paragraph{Step 2: eliminating empty bundles.}
    Among all saturated families maximizing $C(\gamma)$, choose one with minimum coefficient on the empty bundle. Such a family exists because the feasible region is compact (a closed bounded subset of $\mathbb{R}^{2^\ell}$) and the objective is continuous.
    
    Suppose, for contradiction, that the empty bundle has positive coefficient. Then some bundle containing at least two goods must also have positive coefficient. Otherwise, every nonempty subset with positive coefficient would be a singleton. Hence each unit of nonempty coefficient mass would contribute exactly one unit of incidence. Since the empty set has positive coefficient, the total nonempty coefficient mass is strictly smaller than the total mass $p$, and therefore the total incidence would also be strictly smaller than $p$. But saturation gives total incidence $\sum_{g\in G}1=\ell\ge p,$ a contradiction.

    Let $R$ be a subset with at least two goods and positive coefficient $\gamma_R>0$, and choose a good $g\in R$. Since both $\gamma_R$ and the coefficient of the empty set are positive, we may choose $0<\delta\le \min\{\gamma_R,\gamma_{\varnothing}\}.$ We then decrease the coefficients of $R$ and $\varnothing$ by $\delta$, and increase the coefficients of $R-g$ and $\{g\}$ by $\delta$. In other words,
    \[
    \gamma_R\leftarrow\gamma_R-\delta,\qquad
    \gamma_{\varnothing}\leftarrow\gamma_{\varnothing}-\delta,
    \]
    and
    \[
    \gamma_{R-g}\leftarrow\gamma_{R-g}+\delta,\qquad
    \gamma_{\{g\}}\leftarrow\gamma_{\{g\}}+\delta.
    \]
    
    The total coefficient mass is unchanged, since we remove total mass $2\delta$ and add total mass $2\delta$. The incidence of every good is also unchanged:
    To see that every good's incidence is preserved, consider the goods separately. For any good $h\in R-g$, the transformation removes coefficient $\delta$ from the set $R$ but adds the same coefficient $\delta$ to the set $R-g$, and both sets contain $h$. Hence the total incidence of $h$ is unchanged. For the distinguished good $g$, the transformation removes coefficient $\delta$ from the set $R$ but adds coefficient $\delta$ to the singleton $\{g\}$, so the incidence of $g$ is also unchanged. Goods outside $R$ are unaffected.

    Since $v_g\le1<\lambda$,
    \[
    \min\{\lambda,v(R)-v_g\}+\min \{\lambda,v_g\}
    =
    \min\{\lambda,v(R)-v_g\}+v_g
    \ge
    \min\{\lambda,v(R)\}.
    \]
    Indeed, if $v(R)\le\lambda$, both sides equal $v(R)$; if $v(R)>\lambda$, the left-hand side is at least $\lambda$. Thus the transformation does not decrease $C(\gamma)$, while it strictly decreases the coefficient of the empty bundle. This contradicts the choice of the family. Therefore, we may assume that no coefficient is placed on the empty bundle.

\paragraph{Step 3: lower-bounding the singleton mass.}
    Let $w_1$ be the total coefficient mass of singleton bundles. Since every good is saturated, $\ell=\sum_R\gamma_R|R|.$ There are no empty bundles; singleton bundles contribute exactly $w_1$ to this sum, while every nonsingleton bundle has cardinality at least two. Therefore, $\ell\ge w_1+2(p-w_1)=2p-w_1,$ and hence $w_1\ge 2p-\ell.$ Thus the singleton bundles carry at least $2p-\ell$ units of coefficient mass.
    We select exactly $2p-\ell$ units of this mass. If the last singleton bundle
    $\{g\}$ has coefficient $\gamma_{\{g\}}$ but only
    $\alpha<\gamma_{\{g\}}$ more mass is needed, we split it as
    \[
    (\{g\},\gamma_{\{g\}})
    \longrightarrow
    (\{g\},\alpha),\,
    (\{g\},\gamma_{\{g\}}-\alpha),
    \]
    and select only the first copy. This changes neither the total coefficient
    mass nor any good's incidence; it only lets us select the required singleton
    mass exactly.
        
    For a singleton $\{g\}$, its contribution to the capped objective $C(\gamma)$ is $\gamma_{\{g\}}v_g$, 
    because $v_g\le1<\lambda$. Moreover, the incidence constraint implies that the total singleton coefficient assigned to any good $g$ is at most one. Therefore, among all selections of singleton coefficient mass totaling $2p-\ell$, the largest possible contribution to $C(\gamma)$ is obtained by assigning one unit of mass to each of the $2p-\ell$ largest-weight goods. Hence the selected $2p-\ell$ units of singleton mass contribute at most 
    $\sum_{j=1}^{2p-\ell}v_j$ to $C(\gamma)$.

The remaining coefficient mass is $p-(2p-\ell) = \ell-p$, and every bundle contributes at most $\lambda$ to the capped objective. Hence $C(\gamma) \le \sum_{j=1}^{2p-\ell}v_j + (\ell-p)\lambda.$ Combining this with Equation \eqref{ineq}
gives
\[
p\lambda-(\sum_R\gamma_R\bigl(\lambda-v(R)\bigr)_+)
\le
\sum_{j=1}^{2p-\ell}v_j+(\ell-p)\lambda.
\]
Rearranging,
\[
\sum_R\gamma_R\bigl(\lambda-v(R)\bigr)_+ +\sum_{j=1}^{2p-\ell}v_j \ge (2p-\ell)\lambda,
\]
which is \eqref{eq:fractional-count}.
\end{proof}

\begin{remark}[The role of the small-good budget]\label[remark]{rem:small-budget}
Suppose the fractional reference bundles are integral bundles of value at least $\lambda$, and let $G$ denote a distinguished subset of their goods. For a reference bundle $B_j$, $v(B_j)\ge\lambda$ implies $\bigl(\lambda-v(B_j\cap G)\bigr)_+ \le v(B_j\setminus G).$ Thus the deficit left after counting only the distinguished goods can be paid for by the remaining goods in the reference bundle.

After multiplying by the reference coefficients and summing, the per-good incidence bound implies that the total deficit is at most the total value $S$ of the nondistinguished goods. 
\end{remark}

\subsection{LGLB partitions and explicit MXS witnesses}

\begin{lemma}[LGLB removal scores]\label[lemma]{lem:LGLB-score}
Let $\mathbf{C}=(C_1,\ldots,C_q)=\LGLB(v,F,q), \qquad L=\min_{j\in[q]}v(C_j).$ Then $\vefx_v(C_j)\le L \qquad\text{for every }j\in[q].$ \end{lemma}

\begin{proof}
Fix a nonempty bundle $C_j$. Let $g$ be the last good inserted into $C_j$. Because LGLB processes goods in nonincreasing order, $g$ is a minimum-valued good of $C_j$. Hence $\vefx_v(C_j) = v(C_j-g),$ which is exactly the load of $C_j$ immediately before $g$ was inserted.

At that moment, LGLB chose $C_j$ because it had minimum current load. Consider a bundle $C_r$ whose final value is $L$. Its load at that moment was at most its final value $L$. Therefore,
$v(C_j-g)\le L.$
Thus $\vefx_v(C_j)\le L.$ If $C_j$ is empty, its EFX score is zero, so the inequality is immediate.
\end{proof}

\subsection{A certificate with two numerical tests}
\label{subsec:apsx-certificate}
We now define the numerical certificate used to compute the fixed APS floor. Fix one valuation $v$ and a threshold $t>0$. Partition the goods into
\begin{equation}\label{eq:apsx-classes}
H_t=\{g:v(g)\ge t\},\qquad
M_t=\{g:t/2\le v(g)<t\},\qquad
S_t=\{g:v(g)<t/2\}.
\end{equation}
We call these goods high, medium, and small, respectively. Write $h=|H_t|,\qquad \ell=|M_t|.$ For $0\le a\le \ell$, let $T_a$ be the total value of the $a$ most valuable medium goods; in particular, $T_0=0$.

\begin{definition}[Numerically certified threshold]
\label[definition]{def:apsx-certified}
A threshold $t>0$ is \emph{certified} if either $h\ge n$, or $h<n$ and, putting $p=n-h$, the following requirements hold:
\begin{align}
v(M_t\cup S_t)&\ge\frac{3p-1}{2}\,t,\label{eq:apsx-total}\\
v(S_t)+T_a&\ge\frac{3a-1}{2}\,t
\quad\text{if }p<\ell<2p,\qquad a=2p-\ell.
\label{eq:apsx-singleton}
\end{align}
The second requirement is imposed only when $p<\ell<2p$. We also declare zero to be certified.
\end{definition}

The first test controls the total value outside the high goods. The second test is needed when there are more than $p$ but fewer than $2p$ medium goods: it controls the small-good value together with the largest medium goods that may remain unpaired. Both tests use only counts and sums of singleton values.

\begin{lemma}[The target APS scale is certified]
\label[lemma]{lem:apsx-aps}
If $\tau=\APS_{1/n}(v;M)>0$, then $t=\beta_n\tau$ is certified.
\end{lemma}

\begin{proof}
By \Cref{lem:aps-fractional}, there is a fractional reference family $(B_j,\gamma_j)$ satisfying
\[
\sum_j\gamma_j=n,\qquad
\sum_{j:g\in B_j}\gamma_j\le1\quad(g\in M),\qquad
v(B_j)\ge\tau.
\]
If $h\ge n$, certification is immediate. Suppose therefore that $h<n$, and put $p=n-h$.

Discard every reference bundle intersecting $H_t$. By the last assertion of \Cref{lem:aps-fractional}, the discarded reference mass is at most $h$. Thus the remaining mass is at least $p$. Retain exactly $p$ units of this mass, reducing coefficients if necessary. 

Every retained reference is contained in $M_t\cup S_t$. Its value is at least $\tau$, and the incidence of every good is still at most one. Consequently, $v(M_t\cup S_t)\ge\sum_{j} \gamma_jv(B_j)\ge p\tau =pc_nt\ge pc_pt=\frac{3p-1}{2}\,t,$ where the sums now range over the retained family. This proves the first test.

Now suppose that $p<\ell<2p$, and put $a=2p-\ell$.
Project each retained reference onto its medium goods, writing $R_j=B_j\cap M_t.$ Since $v(B_j)\ge\tau$, $\bigl(\tau-v(R_j)\bigr)_+\le v(B_j\cap S_t).$ By additivity and the incidence constraints,
the coefficient-weighted sum of these deficits is at most $v(S_t)$, i.e.,  $\sum_R\gamma_R\bigl(\tau-v(R)\bigr)_+ \leq v(S_t)$.

Apply \Cref{lem:fractional-count} after dividing all values by $t$. The distinguished goods are the $\ell$ medium goods, each of normalized value less than one, the target is $\lambda=\frac{\tau}{t}=c_n>1,$ and the small-good budget is $v(S_t)/t$. The lemma gives $\frac{v(S_t)+T_a}{t}\ge a\frac{\tau}{t}.$ Multiplying by $t$ and using $a<p\le n$, we obtain $v(S_t)+T_a\ge a\tau=ac_nt\ge ac_at=\frac{3a-1}{2}\,t.$ This proves the second test and completes the proof.
\end{proof}

\subsection{Every certified threshold has a residual LGLB splitter}
\label{subsec:apsx-residual}

We next show that the numerical certificate is robust to the deletions required by the allocation engine. The threshold in the following theorem may be any certified value; it need not equal $\beta_n\APS_{1/n}(v;M)$.

\begin{theorem}[Residual splitting from the two tests]
\label{thm:apsx-split}
Let $t>0$ be certified. Let $D_1,\ldots,D_k$ be pairwise disjoint bundles, where $0\le k<n$, such that $v(D_j)<t\qquad(j\in[k]).$ Then every bundle in $\mathbf C=\LGLB\left(v,M\setminus\bigcup_{j=1}^kD_j,n-k\right)$ has value at least $t$.
\end{theorem}

\begin{proof}
We argue by contradiction, after normalizing the threshold.

\paragraph{Normalization.}
Divide all singleton values by $t$. The three classes in \eqref{eq:apsx-classes} are unchanged, and both numerical tests scale by the same factor. Hence it suffices to prove the theorem when $t=1$.

Suppose, for contradiction, that the final LGLB partition has minimum value $L<1$.

Fix a bundle attaining this minimum. Its value at any earlier time is at most its final value $L$. Therefore the minimum current load is at most $L$ at every LGLB insertion.

\paragraph{Removing high goods.}
No high good is deleted, since its value is at least one while every deleted bundle has value below one. LGLB considers all high goods first. Under the supposed failure, each high good enters an empty bundle and that bundle is never selected again: after the insertion its load is at least one, strictly above $L$.

It follows that $h<n-k;$ otherwise every residual bundle would contain a high good and would have value at least one. In particular, $h<n$, so the two numerical tests apply. Remove these $h$ singleton bundles from consideration and put $p=n-h,\qquad q=n-k-h=p-k\ge1.$ The reduced original instance consists of the medium and small goods. These goods occur either in the $q$ remaining LGLB bundles or in the $k$ deleted bundles. Their total value is $v(M_t\cup S_t)$, and the first numerical test gives
\begin{equation}\label{eq:apsx-reduced-total}
v(M_t\cup S_t)\ge\frac{3p-1}{2}.
\end{equation}

Let $\ell_F$ be the number of residual medium goods and let $d$ be the number of deleted medium goods. Thus $\ell=\ell_F+d.$ Since two medium goods have total value at least one, a deleted bundle contains at most one of them. Therefore $d\le k.$ If a residual LGLB bundle receives two medium goods, its value becomes at least one, so it is never selected again. Such a bundle consists of exactly those two goods. Any other residual bundle is empty, or is a singleton of value below one, or has a small good as its last inserted good. In the last case, its final value is below $L+\frac12<\frac32.$ \paragraph{Case 1: $\ell_F\le q$.}
The first $\ell_F$ medium goods enter distinct empty bundles. Thus no residual bundle contains two medium goods. Every residual load is below $3/2$, and at least one has value $L<1$. The deleted total is at most $k$, so $v(M_t\cup S_t)<\frac32(q-1)+1+k =\frac{3p-1}{2}-\frac{k}{2} \le\frac{3p-1}{2}.$ This contradicts \eqref{eq:apsx-reduced-total}.

\paragraph{Case 2: $\ell_F>q$.}
We first show that $q<\ell_F<2q.$ The first $q$ medium goods enter distinct empty bundles. A bundle that has received two medium goods has value at least one, whereas a bundle containing only its initial medium good has value below one. Thus no bundle receives a third medium good before every bundle has received a second one. If $\ell_F\ge2q$, every residual load would be at least one, contradicting $L<1$.

Call a bundle \emph{paired} if it receives two medium goods, and \emph{unpaired} otherwise. There are $\ell_F-q$ paired bundles, so the number of unpaired bundles is $u=q-(\ell_F-q)=2q-\ell_F,\qquad 1\le u<q.$ The initial medium goods of the unpaired bundles are the $u$ largest residual medium goods. Indeed, every subsequent medium good is assigned to a least-valued bundle among those still containing only their initial medium good. Thus the smaller initial goods are paired first. Let $U$ be the total value of the $u$ unpaired medium initial goods.

Let $Z_F$ and $Z_D$ be the total small-good values in the residual and deleted regions, respectively. Since paired bundles already have value at least one, all residual small goods are assigned to unpaired bundles. The minimum-value bundle is unpaired, and every other unpaired bundle has value below $3/2$. Therefore
\begin{equation}\label{eq:unpaired-mass}
U+Z_F<1+\frac32(u-1)=\frac{3u-1}{2}.
\end{equation}

Each deleted bundle containing a medium good has small-good value below $1/2$, and every remaining deleted bundle has total value below one. Hence
\begin{equation}\label{eq:deleted-small}
Z_D\le\frac d2+(k-d)=k-\frac d2.
\end{equation}
Also, $\ell=\ell_F+d<2q+k=2p-k\le2p.$ We distinguish two subcases according to whether $\ell\le p$ or $p<\ell<2p$.

\emph{Subcase 2a: $\ell\le p$.}
The $\ell_F-u$ medium goods in paired bundles each have value below one. Combining their values with the unpaired bundles and the deleted region gives
\begin{align*}
v(M_t\cup S_t)&<\frac{3u-1}{2}+(\ell_F-u)+k\\
 &=\ell_F+\frac u2-\frac12+k\\
 &=2p-k-\frac{u+1}{2},
\end{align*}
where we used $\ell_F=2q-u=2(p-k)-u$.
The assumption $\ell=\ell_F+d\le p$ implies $2(p-k)-u+d\le p, \qquad\text{hence}\qquad u\ge p-2k+d.$ Substituting this bound yields $v(M_t\cup S_t)<\frac{3p-1}{2}-\frac d2\le\frac{3p-1}{2},$ again contradicting \eqref{eq:apsx-reduced-total}.

\emph{Subcase 2b: $p<\ell<2p$.}
Put $a=2p-\ell.$ Using $\ell=\ell_F+d=2(p-k)-u+d$, we obtain $a=u+2k-d.$ Since $d\le k$, this gives $a\ge u+d$. Among all medium goods, only the $d$ deleted medium goods can be larger than the $u$ unpaired initial goods. Therefore the $a$ largest medium goods can be chosen to include all $u$ unpaired initial goods, resolving equal values arbitrarily. Every additional medium good has value below one, so $T_a\le U+(a-u)=U+2k-d.$ Together with \eqref{eq:unpaired-mass} and \eqref{eq:deleted-small}, this implies
\begin{align*}
v(S_t)+T_a
&\le Z_F+Z_D+U+2k-d\\
&<\frac{3u-1}{2}+k-\frac d2+2k-d\\
&=\frac{3(u+2k-d)-1}{2}\\
&=\frac{3a-1}{2}.
\end{align*}
This contradicts the second numerical test, \eqref{eq:apsx-singleton}.

All cases yield a contradiction. Hence every residual bundle has normalized value at least one. Scaling back proves that every original residual bundle has value at least $t$.
\end{proof}

\begin{remark}[The deletion premise is strict]
\label[remark]{rem:apsx-strict}
The theorem uses $v(D_j)<t$, exactly as required by the rejection invariant of the fixed-threshold epoch. No assertion for arbitrary deleted bundles of value equal to $t$ is needed. In particular, the weak boundaries in \eqref{eq:apsx-classes} cause no loss in the APS coefficient.
\end{remark}

\subsection{Exact polynomial computation of the maximum certified threshold}
\label{subsec:apsx-optimization}

Define $\widehat t(v,n)=\max\bigl(\{0\}\cup\{t>0:t\text{ is certified}\}\bigr).$ The next lemma proves both that this maximum is attained and that it can
be computed exactly, rather than approximated by binary search.

\begin{lemma}[An exact event sweep]
\label[lemma]{lem:apsx-optimize}
The maximum certified threshold, $\widehat t(v,n)$, is computable in $O((m+1)^2)$.  For rational input it has
polynomial encoding length, and
\begin{equation}\label{eq:apsx-floor-range}
 \beta_n\APS_{1/n}(v;M)\le\widehat t(v,n)\le\RMMS_n(v;M) \leq MMS_n(v;M).
\end{equation}
\end{lemma}
\begin{proof}
Put $B=v(M)/n$.  If $B=0$, return zero.  Every positive certified threshold
is at most $B$.  Indeed, when $h\ge n$, the high goods alone have value at
least $nt$.  Otherwise $v(M_t\cup S_t)\ge pc_p t\ge pt$ and the high goods have value
at least $ht$, so again $v(M)\ge(h+p)t=nt$.

Sort the numbers in
\begin{equation}\label{eq:apsx-events}
 \mathcal E=\{0,B\}\cup
 \bigl(\{v(g),2v(g):g\in M\}\cap(0,B)\bigr)
 =\{e_0<e_1<\cdots<e_r\},
\end{equation}
and call every number an event.
There are at most $2m+2$ events.  For all $t \in  (e_{j-1},e_j]$, all
three classes in \eqref{eq:apsx-classes} are constant and equal to when $t = e_j$.  This includes the right endpoint because equality at
$t=v(g)$ keeps the good in $H_t$, and equality at $t=2v(g)$ keeps it in
$M_t$.

Compute $h,\ell,v(M_t\cup S_t),v(S_t)$ and the sorted medium-good values for that interval.
If $h\ge n$, set $R_j=+\infty$.  Otherwise put $p=n-h$ and set
\begin{equation}\label{eq:apsx-root}
 R_j=
 \begin{cases}
 \displaystyle\min\left\{\frac{2v(M_t \cup S_t)}{3p-1},
                  \frac{2(v(S_t)+T_{2p-\ell})}{3(2p-\ell)-1}\right\},
                     &p<\ell<2p,\\[1.0em]
 \displaystyle\frac{2v(M_t \cup S_t)}{3p-1},&\text{otherwise}.
 \end{cases}
\end{equation}
Exactly the thresholds $t\in(e_{j-1},e_j]$ with $t\le R_j$ are certified.
Thus this interval contributes the candidate
\begin{equation}\label{eq:apsx-candidate}
 t_j=\min\{e_j,R_j\}\quad\text{if }\min\{e_j,R_j\}>e_{j-1},
\end{equation}
and contributes no candidate otherwise.  Return the maximum of zero and
all contributed candidates.  Every nonempty feasible interval has its
largest point included, so this returns the claimed maximum.  No global
monotonicity of certification is assumed.

Sort the goods once by value.  For each of $O(m+1)$ intervals, a scan of
this order computes the classes, sums, and the required largest-good
sum in $O(m+1)$ work.  Event sorting has smaller cost.  This proves the
operation bound.  All positive candidates are event values or twice an
original subset sum divided by one of the integers $3p-1$ or $3a-1$;
the event $B$ additionally uses the denominator $n$.  Therefore their
encoding lengths are polynomial in the rational input length.

For the lower bound in \eqref{eq:apsx-floor-range}, apply
\Cref{lem:apsx-aps} when APS is positive; when it is zero the claim is
immediate.  For the upper bounds, \Cref{thm:apsx-split} shows that every positive certified $t$ is residual self-feasible in the sense of Definition~\ref{def:rmms}. Hence $t\le\RMMS_n(v;M)$. Taking $k=0$ also supplies an actual $n$-partition of minimum value at least $t$, and the definition of RMMS gives $\RMMS_n(v;M)\le\MMS_n(v;M)$. Zero satisfies the same bounds. This also proves \Cref{theorem:2/3aps bounded}.
\end{proof}

\textbf{Proof of \cref{theorem:2/3aps bounded}.} Lemmas \cref{thm:apsx-split} and \cref{lem:apsx-optimize} complete the proof.

\subsection{The strengthened joint residual certificate}

    The next theorem does not assume that the deleted bundles are below any prescribed threshold. It produces an MXS upper witness and shows that its value cannot fall below $\widehat t(v,n)$ computed above.

\begin{theorem}[Joint residual MXS--APS certificate]\label{thm:joint-certificate}
Fix one nonnegative additive valuation $v$ and pairwise disjoint bundles
$D_1,\ldots,D_k$, where $0\le k<n$. Let
\[
\mathbf{C}
=
(C_1,\ldots,C_{n-k})
=
\LGLB\left(v,M\setminus\bigcup_{j=1}^kD_j,n-k\right),
\qquad
L=\min_jv(C_j),
\]
and define
\begin{equation}\label{eq:joint-b}
 b=\max\{L,v(D_1),\ldots,v(D_k)\}.
\end{equation}
Then
\begin{equation}\label{eq:joint-conclusion}
 \max\{\MXS_n(v;M),\widehat t(v,n)\}\le b,
 \qquad
 \APS_{1/n}(v;M)\le c_n b.
\end{equation}
Moreover, the $n$-partition $\mathbf{B} = (D_1,\ldots,D_k,C_1,\ldots,C_{n-k})$ contains an explicitly identifiable EFX-feasible bundle of value exactly $b$. Hence $\mathbf{B}$ together with this distinguished bundle is an MXS upper witness of value $b$.
\end{theorem}

\begin{proof}
We first prove the MXS statement.

For every deleted bundle $D_j$, $\vefx_v(D_j)\le v(D_j)\le b.$ By \Cref{lem:LGLB-score}, every residual LGLB bundle satisfies $\vefx_v(C_j)\le L\le b.$ Thus every bundle in $\mathbf{B}$ has EFX score at most $b$.

If some deleted bundle satisfies $v(D_j)=b,$ choose that bundle as the distinguished bundle. Otherwise, by the definition of $b$, we must have $b=L$; choose any residual bundle $C_r$ with $v(C_r)=L=b.$ In either case, the distinguished bundle has value exactly $b$ and dominates the EFX score of every bundle in $\mathbf{B}$. It is therefore EFX-feasible in $\mathbf{B}$. By the definition of MXS, $\MXS_n(v;M)\le b.$ This also proves the explicit-witness assertion.

It remains to prove that $b\ge\widehat t(v,n)$. If $\widehat t(v,n)=0$, this is immediate. Otherwise, suppose that $b<\widehat t(v,n).$ Then every deleted bundle has value below $\widehat t(v,n)$, and so does the minimum residual LGLB bundle. This contradicts \Cref{thm:apsx-split}, because $\widehat t(v,n)$ is certified by \Cref{lem:apsx-optimize}.

Finally, \Cref{lem:apsx-optimize} gives $\beta_n\APS_{1/n}(v;M)\le\widehat t(v,n)\le b.$ Since $\beta_n=1/c_n$, this implies $\APS_{1/n}(v;M)\le c_nb$ and proves the joint statement.
\end{proof}

\begin{corollary}[Initialization interval]\label[corollary]{cor:initial-range}
For agent $i$, let
\[
t_i^*=\widehat t(v_i,n),\qquad
\mathbf C^{(0)}=(C_1^{(0)},\ldots,C_n^{(0)})=\LGLB(v_i,M,n),
\]
and let $b_i^{(0)}=\min_{j\in[n]}v_i(C_j^{(0)}).$ Then
\begin{equation}\label{eq:initial-range}
\max\{\xi_i,t_i^*\}
\le b_i^{(0)}\le\mu_i\le\tau_i,
\qquad t_i^*\ge\beta_n\tau_i.
\end{equation}
In particular, $b_i^{(0)}=0\quad\Longleftrightarrow\quad\tau_i=0.$ \end{corollary}

\begin{proof}
Apply \Cref{thm:joint-certificate} with $k=0$. In this case $b=b_i^{(0)}$, so the theorem gives $\max\{\xi_i,t_i^*\}\le b_i^{(0)}.$ The lower bound $t_i^*\ge\beta_n\tau_i$ is \eqref{eq:apsx-floor-range}. The LGLB partition is one particular $n$-partition of $M$, so its minimum bundle value cannot exceed MMS: $b_i^{(0)}\le\MMS_n(v_i;M)=\mu_i.$ Finally, \Cref{sec:prelim} gives $\mu_i\le\tau_i$, proving \eqref{eq:initial-range}.

If $b_i^{(0)}=0$, then $0\le\beta_n\tau_i\le t_i^*\le b_i^{(0)}=0$. Since $\beta_n>0$, this implies $\tau_i=0$. Conversely, if $\tau_i=0$, then $0\le b_i^{(0)}\le\mu_i\le\tau_i=0$, so $b_i^{(0)}=0$.
\end{proof}

\begin{remark}[The two certificates serve different purposes]
\label[remark]{rem:apsx-scope}
The fixed floor $t_i^*$ is residual self-feasible and satisfies $t_i^*\ge\beta_n\tau_i$, but it need not dominate MXS. The value $b_i$ has an explicit MXS upper witness and remains above $t_i^*$, but it is not asserted to be residual self-feasible. Keeping these roles separate allows \Cref{sec:aps-algorithms} to retain the fixed APS floor while lowering MXS witnesses only when an epoch supplies a new one.
\end{remark}

\section{Allocation at fixed thresholds}\label{sec:engine}

The purpose of this section is to develop the common allocation engine used by our later algorithms. Fix a threshold $t_i\ge0$ for every agent $i$, together with a partition routine $\mathsf{Divide}$. Starting from the empty allocation, the engine repeatedly performs divide, trim, repair, and matching steps until it either assigns every agent a bundle meeting her threshold or finds a partition whose minimum value is below the divider's threshold.

We call one complete execution of this process, from its initialization until it returns either success or failure, an \emph{epoch}. An epoch starts from the partial allocation
\[
\mathbf{A}=\langle A_1,\ldots,A_n\rangle
\quad
\text{with}
\quad
A_i=\varnothing
\qquad
\text{for every }i\in N,
\]
free pool $F=M$, assigned-agent set $H=\{i\in N:t_i=0\},$ and unassigned-agent set $P=N\setminus H.$ Thus the agents with zero threshold are regarded as assigned from the beginning, while every agent in $P$ has a positive threshold. Call a bundle $B\subseteq F$ \emph{desired} if $v_i(B)\ge t_i$ for at least one unassigned agent $i\in P$.

Throughout an epoch we maintain the following three invariants:
\begin{equation}\label{eq:engine-own}
v_i(A_i)\ge t_i
\qquad
\text{for every }i\in H.
\end{equation}
\begin{equation}\label{eq:engine-reject}
v_i(A_j)<t_i
\qquad
\text{for every }i\in P
\text{ and }j\in H.
\end{equation}
\begin{equation}\label{eq:engine-efl}
\mathbf{A}\text{ is EFL among the agents in }H.
\end{equation}

We refer to \eqref{eq:engine-reject} as the \emph{rejection invariant}. It says that every currently unassigned agent values every currently assigned bundle strictly below her own threshold. This strict inequality will later ensure that a failed-division update strictly decreases the divider's threshold.

The structure of an epoch is as follows. In each round, an unassigned agent divides the currently free goods into as many bundles as there are unassigned agents. If the divider values some bundle below her threshold, the epoch stops and returns this residual partition as a failure certificate. Otherwise, each bundle is trimmed to an inclusion-minimal bundle desired by some unassigned agent.

At that point, if some assigned agent EFL-envies one of the candidate bundles with respect to her current bundle, \Cref{alg:repair} replaces one assigned bundle by a strictly better proper subset of that candidate. \Cref{lem:repair} proves that this operation preserves all three invariants, and \Cref{lem:breakpoints} shows that only polynomially many repairs can occur. If no assigned agent EFL-envies any candidate, \Cref{lem:matching} finds a nonempty matching of candidates to unassigned agents that preserves the rejection invariant, and \Cref{lem:admission} proves that allocating these matched candidates preserves EFL.

The complete procedure is given in \Cref{alg:epoch}, and its guarantee is summarized in \Cref{prop:epoch}. Finally, when all agents are assigned, \Cref{lem:completion} shows how to allocate all remaining free goods without decreasing any agent's value or destroying EFL.

\subsection{Minimal desired candidates}
At the start of an epoch, 
choose any divider $d\in P$. Note that an epoch gets an algorithm $\mathsf{Divide}$ as input. 
The partition routine returns a $|P|$-partition $\mathbf{C}=(C_1,\ldots,C_{|P|}) = \mathsf{Divide}(v_d,F,|P|)$ of the current free pool $F$. Define $L=\min_{j\in[|P|]}v_d(C_j).$ If $L<t_d$, the epoch terminates and reports the divider $d$, the current allocation $\mathbf{A}$, the residual partition $\mathbf{C}$, and its minimum value $L$. Later sections will use this failed division to update the divider's threshold.

Suppose instead that $L\ge t_d$. Then every bundle $C_j$ is desired by the divider. 

Starting from each $C_j$, repeatedly delete goods as long as the remaining bundle is still desired. Denote the resulting bundle by $B_j$. Thus $B_j$ is an \emph{inclusion-minimal desired candidate}: it is desired by at least one unassigned agent, but none of its proper subsets is desired by any unassigned agent.

Only one scan through the goods is needed to obtain such a candidate. Indeed, suppose that at some stage the current bundle is $B$ and deleting a good $g\in B$ makes $B-g$ undesired. Then $v_i(B-g)<t_i \qquad \text{for every }i\in P.$ If further goods are subsequently deleted while $g$ remains, the current bundle becomes some $B'\subseteq B$ with $g\in B'$. Then $B'-g\subseteq B-g$. 
By monotonicity, $v_i(B'-g)\le v_i(B-g)<t_i \qquad \text{for every }i\in P.$ Thus deleting $g$ can never become possible later, so there is no need to reconsider it.

Consequently, every candidate $B_j$ satisfies
\begin{equation}\label{eq:minimal-candidate}
v_i(S)<t_i
\qquad
\text{for every }i\in P
\text{ and every }S\subsetneq B_j.
\end{equation}

Every candidate is nonempty because all agents in $P$ have positive thresholds. The candidates remain pairwise disjoint because they are obtained by deleting goods from the bundles of the partition $\mathbf{C}$. Every good removed during trimming remains in the free pool. Notice also that, although the original bundle $C_j$ is desired by the divider, the final candidate $B_j$ need not be: another unassigned agent may be the one whose threshold keeps $B_j$ desired.

\subsection{An EFL-preserving repair with polynomially many utility levels}

Suppose that some assigned agent $i\in H$ EFL-envies a candidate bundle $B$ with respect to her current bundle $A_i$. By definition, this means $v_i(A_i)<\vefl_i(B).$ The repair procedure will replace one currently assigned bundle by a proper subset of $B$.

The requirement that the replacement be a \emph{proper} subset is important. By \eqref{eq:minimal-candidate}, every unassigned agent values every proper subset of $B$ strictly below her threshold. Hence giving such a subset to an already assigned agent preserves the rejection invariant.

\begin{algorithm}[t]
\caption{$\Repair(B,H,\mathbf{A})$: repair an EFL-envied candidate}\label{alg:repair}
\begin{algorithmic}[1]
\Require A candidate bundle $B$, the current assigned-agent set $H$, and the current partial allocation $\mathbf{A}=\langle A_1,\ldots,A_n\rangle$ such that some $i\in H$ EFL-envies $B$ with respect to $A_i$
\State Set $u_i\gets v_i(A_i)$ for every $i\in H$
\If{some $i\in H$ and $g\in B$ satisfy $v_i(g)>u_i$}
    \State \Return $(i,\{g\})$
\EndIf
\State $S\gets B$; $a\gets\bot$
\While{some $i\in H$ satisfies $\vefone_i(S)>u_i$}
    \State Choose such an $i$ and an $i$-minimum good $g\in S$
    \State $S\gets S-g$; $a\gets i$
\EndWhile
\State \Return $(a,S)$
\end{algorithmic}
\end{algorithm}

The next lemma shows that the returned bundle strictly improves its recipient and is not EFL-envied by any assigned agent relative to her utility before the repair.

\begin{lemma}[Correctness of EFL repair]\label{lem:repair}
Algorithm~\ref{alg:repair} returns an assigned recipient $a\in H$ and a proper subset $S\subsetneq B$ such that $v_a(S)>v_a(A_a),$ and, for every $i\in H$, $\vefl_i(S)\le v_i(A_i),$ where the right-hand side is the utility of agent $i$ before the repair. Equivalently, no assigned agent EFL-envies $S$ with respect to her old bundle.

Replacing $A_a$ by $S$ and releasing the old bundle $A_a$ preserves \eqref{eq:engine-own}--\eqref{eq:engine-efl}.
\end{lemma}

\begin{proof}
For each assigned agent $i\in H$, let $u_i=v_i(A_i)$ denote her utility at the beginning of the repair call. These values are fixed throughout the procedure.

We consider the two branches separately.

\medskip
\noindent
\textbf{Singleton branch.}
Suppose there are an agent $i\in H$ and a good $g\in B$ such that $v_i(g)>u_i$. 
The algorithm returns $(i,\{g\})$.

Since some assigned agents EFL-envies the input candidate we get $|B|\ge2$, the singleton $\{g\}$ is a proper subset of $B$. Moreover, $v_i(g)>u_i=v_i(A_i),$ so the recipient strictly improves.

For every assigned agent $j$, $\vefl_j(\{g\})=0.$ Since $u_j\ge0$, $\vefl_j(\{g\})\le u_j.$ Thus no assigned agent EFL-envies the singleton $\{g\}$ with respect to her old bundle.

\medskip
\noindent
\textbf{Deletion branch.}
Suppose the singleton branch does not apply. Then $v_i(g)\le u_i \qquad \text{for every }i\in H \text{ and every }g\in B.$ In particular, this inequality holds for every good in every subset $S\subseteq B$ encountered by the while-loop.

Fix an assigned agent $i$ and such a subset $S$. We claim that $\vefl_i(S)\le u_i \quad\Longleftrightarrow\quad \vefone_i(S)\le u_i.$ To see this, recall that $\vefl_i(S) = \min_{g\in S} \max\{v_i(g),v_i(S-g)\}.$ Because $v_i(g)\le u_i \qquad \text{for every }g\in S,$ the inequality $\vefl_i(S)\le u_i$ holds exactly when there exists a good $g\in S$ such that $v_i(S-g)\le u_i.$ Among all goods $g\in S$, the quantity $v_i(S-g)$ is minimized by deleting an $i$-maximum good. Therefore $\min_{g\in S}v_i(S-g) = v_i(S)-\max_{g\in S}v_i(g) = \vefone_i(S).$ Hence $\vefl_i(S)\le u_i \quad\Longleftrightarrow\quad \vefone_i(S)\le u_i,$ as claimed.

The input candidate $B$ is EFL-envied by some assigned agent. Therefore, for some $i\in H$, $\vefl_i(B)>u_i.$ By the equivalence above, $\vefone_i(B)>u_i.$ Thus the while-loop executes at least once, so its output $S$ is a proper subset of $B$.

The loop must eventually terminate. Indeed, once $S$ is a singleton, $\vefone_i(S)=0 \qquad \text{for every }i\in H,$ so the loop condition cannot hold. Therefore at most $|B|-1$ deletions occur.

When the loop terminates, $\vefone_i(S)\le u_i \qquad \text{for every }i\in H.$ Using the equivalence above once more, $\vefl_i(S)\le u_i \qquad \text{for every }i\in H.$ Hence no assigned agent EFL-envies the returned bundle $S$ with respect to her old bundle.

It remains to show that the recipient $a$'s value strictly improves.

Consider the final iteration of the while-loop. Let $T$ be the bundle immediately before the good is deleted in that iteration. Agent $a$ is the agent whose EF1 inequality is violated at that moment. Under $v_a$, define $V=v_a(T), \qquad x=\min_{g\in T}v_a(g), \qquad z=\max_{g\in T}v_a(g).$ Because agent $a$ triggers the loop, $\vefone_a(T)>u_a.$ Since $\vefone_a(T)=V-z,$ we obtain $u_a<V-z.$ The algorithm deletes an $a$-minimum good, whose value is $x$. Because this is the final iteration, the resulting bundle is precisely the returned bundle $S$, and therefore $v_a(S)=V-x.$ Since $x\le z$,
\begin{equation}\label{eq:repair-gain}
u_a<V-z\le V-x=v_a(S).
\end{equation}
Thus the recipient $a$ strictly improves.

We finally verify the three epoch invariants.

First, agent $a$ met her threshold before the repair and strictly increases her utility, so $v_a(S)>v_a(A_a)\ge t_a.$ Every other assigned agent keeps her own bundle. Hence \eqref{eq:engine-own} is preserved.

Second, for every previously assigned agent $i$, $\vefl_i(S)\le u_i=v_i(A_i)$ before the replacement, and the recipient's own utility only increases. Thus no assigned agent EFL-envies the new assigned bundle $S$. The old bundle $A_a$ is released and therefore ceases to be an assigned rival bundle. All other assigned bundles are unchanged. Hence \eqref{eq:engine-efl} is preserved.

Third, because $S\subsetneq B$, the minimality property \eqref{eq:minimal-candidate} gives $v_i(S)<t_i \qquad \text{for every }i\in P.$ All other assigned bundles remain unchanged, so \eqref{eq:engine-reject} is preserved.
\end{proof}

The previous lemma guarantees that every repair strictly increases the recipient's utility. This alone does not imply a polynomial bound on the number of repairs, since the utility increase could be arbitrarily small. The next lemma shows that each repair crosses one of only polynomially many (with respect to $m$ and independently of the numerical magnitudes or gaps between values) fixed utility levels.

\begin{lemma}[Fixed-breakpoint bound]\label{lem:breakpoints}
There are $O(nm^4)$ repairs in any epoch.
\end{lemma}

\begin{proof}
Fix an agent $i$. Define
\[
\mathcal Q_i
=
\{r\,v_i(g)+s\,v_i(h):
g,h\in M,\ 
r,s\in\mathbb Z_{\ge0},
r+s\le m-1\}
\cup
\{v_i(g):g\in M\}.
\]
The set $\mathcal Q_i$ is used only in the analysis; the algorithm never needs to construct or enumerate it.

There are $m^2$ choices for the ordered pair $(g,h)$. There are $O(m^2)$ nonnegative integer pairs $(r,s)$ satisfying $r+s\le m-1$. 
Therefore $|\mathcal Q_i|=O(m^4).$ We show that every time agent $i$ receives a repaired bundle, her utility crosses a breakpoint in $\mathcal Q_i$ that she has never crossed before.

First suppose the repair returns a singleton $\{g\}$. By the condition of the singleton branch, $v_i(A_i)<v_i(g)$. Since $v_i(g)\in\mathcal Q_i$, 
the repair raises agent $i$'s utility from strictly below this breakpoint to the breakpoint itself.

Now consider a repair produced by the deletion branch. Let $T$ be the bundle immediately before the final deletion, and suppose $T$ contains $k+1$ goods. Define $V=v_i(T), \qquad x=\min_{g\in T}v_i(g), \qquad z=\max_{g\in T}v_i(g).$ Because $T$ violates the EF1 inequality for agent $i$, her utility before the repair satisfies $v_i(A_i)<V-z$. 
The algorithm deletes an $i$-minimum good, so her utility after the repair is $V-x$. 
Hence it suffices to show that the interval $I=[V-z,V-x]$ contains a point of $\mathcal Q_i$.

The lower endpoint $V-z$ is the value of the $k$ goods left after deleting an $i$-maximum good. Every one of these goods has value at least $x$, so $V-z\ge kx.$ Similarly, $V-x$ is the value of the $k$ goods left after deleting an $i$-minimum good. Every one of them has value at most $z$, so $V-x\le kz.$ Therefore $I\subseteq[kx,kz].$ Moreover, $(V-x)-(V-z)=z-x,$ so the length of $I$ is exactly $z-x$.

If $x=z$, then $V-z=V-x=kx,$ and hence $I=\{kx\}.$ The point $kx$ belongs to $\mathcal Q_i$.

Suppose now that $x<z$. Consider the grid $kx,\ kx+(z-x), kx+2(z-x), \ldots, kx+k(z-x)=kz.$ The distance between two consecutive grid points is exactly $z-x$, the length of $I$. Since $I\subseteq[kx,kz]$, the interval $I$ must contain at least one grid point.

More explicitly, let $\zeta$ be the smallest grid point satisfying $\zeta\ge V-z.$ Then $\zeta\le(V-z)+(z-x)=V-x,$ so $\zeta\in I.$ Every grid point has the form $(k-j)x+jz$ for some integer $0\le j\le k$. Since $k\le m-1$, this point belongs to $\mathcal Q_i$.

Thus every non-singleton repair of agent $i$ crosses some $\zeta\in\mathcal Q_i$ satisfying $v_i(A_i)<\zeta\le v_i(S).$ During an epoch, an assigned agent never becomes unassigned, and every repair only increases the recipient's utility. Consequently, once agent $i$ has reached or passed a breakpoint $\zeta$, she can never cross $\zeta$ again: another crossing would require her utility immediately before a later repair to be strictly below $\zeta$.

Hence agent $i$ can receive at most $|\mathcal Q_i|=O(m^4)$ repairs. Summing over all $n$ agents gives $O(nm^4)$ repairs in one epoch.
\end{proof}

\subsection{A matching that preserves rejection}

Suppose now that no assigned agent EFL-envies any candidate bundle with respect to her current bundle. We then try to allocate some of the candidates to unassigned agents.

Construct a bipartite graph whose left vertices are the unassigned agents in $P$ and whose right vertices are the $|P|$ candidate bundles. Put an edge between agent $i\in P$ and candidate $B_j$ exactly when $v_i(B_j)\ge t_i.$ Thus $i$ is adjacent to $B_j$ exactly when $B_j$ meets agent $i$'s threshold. Every right vertex has at least one neighbor because every candidate is desired by some unassigned agent.

We need more than an arbitrary nonempty matching. If a candidate is allocated, every agent who remains unassigned must value that candidate strictly below her threshold; otherwise the rejection invariant would fail. The next lemma provides a matching with exactly this property.
\begin{lemma}[Nonempty rejection-preserving matching]\label{lem:matching}
A bipartite graph with the same number of vertices in each side with at least one vertex per side and no isolated right vertex has a nonempty matching such that no unmatched left vertex has an edge to a matched right vertex. It can be found in polynomial time.
\end{lemma}

\begin{proof}
    Let $\mathcal M$ be a maximum-cardinality matching.
    
    If $\mathcal M$ is perfect, return $\mathcal M$. In this case there are no unmatched left vertices, so the required property holds vacuously.
    
    Assume therefore that $\mathcal M$ is not perfect. 
    Start an alternating search simultaneously from all unmatched left vertices. From a left vertex, traverse only edges that do not belong to $\mathcal M$. From a right vertex, traverse its unique matching edge in $\mathcal M$, if it has one. Let $X$ be the set of reached left vertices and $Y$ the set of reached right vertices.
    
    No unmatched right vertex can belong to $Y$. Indeed, if an unmatched right vertex were reached, the alternating path from an unmatched left starting vertex to that right vertex would be an augmenting path for $\mathcal M$, contradicting the maximality of $\mathcal M$.
    
    Now retain exactly those edges of $\mathcal M$ whose left endpoint lies outside $X$. Denote the resulting matching by $\mathcal M'$.
    
    We first identify its unmatched left vertices.
    
    Every left vertex that is unmatched in $\mathcal M$ belongs to $X$ by construction, since these are precisely the starting vertices of the search. Therefore every left vertex outside $X$ is matched in $\mathcal M$.
    
    If a left vertex $\ell\notin X$ is matched in $\mathcal M$ to a right vertex $r$, then $r\notin Y$. Otherwise, the alternating search would traverse the matching edge $r\ell$ and would reach $\ell$, contradicting $\ell\notin X$. Hence the matching edge incident to every left vertex outside $X$ is retained in $\mathcal M'$.
    
    On the other hand, every matching edge whose left endpoint lies in $X$ is deleted. Thus the unmatched left vertices of $\mathcal M'$ are exactly the vertices in $X$.
    
    We next prove the rejection property. Suppose, toward a contradiction, that some unmatched left vertex $\ell\in X$
    has an edge to a right vertex $r$ that is matched by $\mathcal M'$. 
    
    Because $r$ is matched in $\mathcal M'$, its matched left endpoint lies outside $X$, and therefore $r\notin Y$.

    Since $r$ is matched in $\mathcal M'$, we get that $r$ and $\ell$ are not matched in $\mathcal M$, so
    the edge $\ell r$ does not belong to $\mathcal M$; therefore, the alternating search would traverse it from $\ell$ to $r$, giving $r\in Y$, a contradiction.
    Therefore no unmatched left vertex has an edge to a right vertex matched by $\mathcal M'$.
    
    It remains to show that $\mathcal M'$ is nonempty.
    
    Suppose instead that $\mathcal M'$ contains no edge. Then every left vertex matched by $\mathcal M$ belongs to $X$. Every left vertex unmatched by $\mathcal M$ also belongs to $X$ because it was a starting vertex. Hence $X$
    contains every left vertex.
    
    Because the graph is square and $\mathcal M$ is not perfect, there is at least one unmatched right vertex $r$. By assumption, no right vertex is isolated, so $r$ has some neighbor $\ell$. Since every left vertex belongs to $X$, we have
    $\ell\in X$.
    The edge $\ell r$ is not a matching edge because $r$ is unmatched. The alternating search would therefore traverse $\ell r$ and reach $r$. This would give an augmenting path from an unmatched left vertex to the unmatched right vertex $r$, contradicting maximality of $\mathcal M$.
    
    Hence $\mathcal M'$ is nonempty and satisfies the required property.
    
    Maximum matching and alternating reachability can both be computed in polynomial time.
\end{proof}

We call the following operation an \emph{admission round}. Compute the matching guaranteed by \Cref{lem:matching}; allocate every matched candidate bundle to its matched unassigned agent; add all matched agents to the assigned set $H$; and remove the allocated goods from the free pool $F$. All unmatched agents remain unassigned, and all unmatched candidate bundles are discarded as temporary objects.

The next lemma shows that an admission round preserves all epoch invariants.

\begin{lemma}[Correctness of an admission round]\label{lem:admission}
An admission round preserves \eqref{eq:engine-own}--\eqref{eq:engine-efl} and increases $|H|$ by at least one.
\end{lemma}

\begin{proof}
Let $\mathcal M'$ be the matching returned by \Cref{lem:matching}. Since $\mathcal M'$ is nonempty, at least one previously unassigned agent becomes assigned. Therefore $|H|$ strictly increases.

We verify the three invariants separately.

\medskip
\noindent
\textbf{Threshold satisfaction.}
If agent $i$ is matched to candidate $B$, then the edge $iB$ exists, and hence $v_i(B)\ge t_i$. 
Thus every newly assigned agent receives a bundle meeting her threshold. Previously assigned agents keep their own bundles. Therefore \eqref{eq:engine-own} is preserved.

\medskip
\noindent
\textbf{EFL among assigned agents.}
We must consider three kinds of comparisons.

First, consider an old assigned agent $i$ and a newly allocated candidate $B$. We enter the matching branch only when no assigned agent EFL-envies any candidate. Thus, before the allocation, $v_i(A_i)\ge\vefl_i(B)$. 
The old assigned agent keeps her own bundle, so this comparison remains valid after $B$ is allocated.

Second, consider a newly assigned agent $i$ and an old assigned bundle $A_j$. Before the admission round, $i$ belonged to $P$. By the rejection invariant, $v_i(A_j)<t_i$. 
The candidate allocated to $i$ becomes her own bundle $A_i$ and satisfies $v_i(A_i)\ge t_i$. 
Hence $v_i(A_i)>v_i(A_j)\ge\vefl_i(A_j).$ Thus $i$ does not EFL-envy any old assigned bundle.

Third, consider two newly assigned agents. Suppose agent $i$ receives candidate $B_r$, while another agent receives candidate $B_s$. By \eqref{eq:minimal-candidate}, every proper subset of $B_s$ has value below $t_i$ for agent $i$. In particular, for every $g\in B_s$, $v_i(B_s-g)<t_i.$ Therefore $\vefx_i(B_s) = \max_{g\in B_s}v_i(B_s-g) <t_i.$ Since agent $i$ receives a desired candidate, $v_i(B_r)\ge t_i.$ It follows that $v_i(B_r)>\vefx_i(B_s).$ Thus agent $i$ does not EFX-envy (and therefore EFL-envy) $B_s$.

Therefore \eqref{eq:engine-efl} is preserved.

\medskip
\noindent
\textbf{Rejection by agents that remain unassigned.}
Let $i$ be an agent who remains unmatched after the admission round. By the property of the matching in \Cref{lem:matching}, $i$ has no edge to any matched candidate $B$. By the definition of the graph, this means $v_i(B)<t_i$ for every newly allocated candidate $B$.

Every old assigned bundle already satisfied \eqref{eq:engine-reject}. Hence every still-unassigned agent values every assigned bundle strictly below her threshold, and the rejection invariant is preserved.

All three invariants therefore hold after the admission round.
\end{proof}

\begin{algorithm}[t]
\caption{$\Epoch(t,\mathsf{Divide})$: one fixed-threshold EFL epoch}\label{alg:epoch}
\begin{algorithmic}[1]

\Require A threshold vector $t=(t_1,\ldots,t_n)$ and a partition routine $\mathsf{Divide}$
\State $A_i\gets\varnothing$ for every $i\in N$; $F\gets M$
\State $H\gets\{i:t_i=0\}$; $P\gets N\setminus H$
\While{$P\ne\varnothing$}
    \State Choose $d\in P$
    \State $\mathbf{C}=(C_1,\ldots,C_{|P|})\gets\mathsf{Divide}(v_d,F,|P|)$
    \State $L\gets\min_jv_d(C_j)$
    \If{$L<t_d$}
        \State \Return $\mathsf{Failure}(d,\mathbf{A},\mathbf{C},L)$
    \EndIf
    \State Trim each $C_j$ to an inclusion-minimal candidate $B_j$ desired by some $i\in P$
    \If{some $i\in H$ EFL-envies some candidate $B_j$ with respect to $A_i$}
        \State $(a,S)\gets\Repair(B_j,H,\mathbf{A})$
        \State $F\gets(F\setminus S)\cup A_a$; $A_a\gets S$
        \State Discard all temporary candidate bundles
    \Else
        \State Find the matching from \Cref{lem:matching} and allocate its candidate bundles
        \State Remove the allocated goods from $F$ and add their recipients to $H$
        \State $P\gets N\setminus H$
    \EndIf
\EndWhile
\State \Return $\mathsf{Success}(\mathbf{A},F)$
\end{algorithmic}
\end{algorithm}

\begin{proposition}[What an epoch guarantees]\label{prop:epoch}
An epoch returns either a partial EFL allocation meeting every threshold $t_i$, or a failure consisting of a divider $d$, the current allocation, and a residual partition whose minimum value $L$ satisfies $L<t_d,$ while every currently assigned bundle has value strictly below $t_d$ for the divider.

An epoch contains at most $n$ admission rounds and $O(nm^4)$ repairs. Apart from calls to $\mathsf{Divide}$, all work is polynomial in $n$ and $m$, independently of the numerical magnitudes of the valuations.
\end{proposition}

\begin{proof}
We first verify that the three invariants hold at the beginning of an epoch.

Every initially assigned agent $i\in H$ has $t_i=0$ and receives $A_i=\varnothing$. 
Hence $v_i(A_i)=0=t_i$, 
so \eqref{eq:engine-own} holds.

Every initially unassigned agent $i\in P$ has positive threshold $t_i>0$. 
Every initially assigned bundle is empty, so for every $j\in H$, $v_i(A_j)=0<t_i$. 
Thus the rejection invariant \eqref{eq:engine-reject} also holds.

Finally, all initially assigned bundles are empty. Therefore the allocation is trivially EFL among the agents in $H$, and \eqref{eq:engine-efl} holds.

Now consider a round in which the algorithm does not terminate with failure. 
If the repair branch is taken, \Cref{lem:repair} shows that all three invariants are preserved. 
If the admission branch is taken, \Cref{lem:admission} shows that all three invariants are preserved and that $|H|$ increases by at least one.

By \Cref{lem:breakpoints}, at most $O(nm^4)$ 
repair rounds can occur in one epoch. Every admission round adds at least one agent to $H$, and no agent ever leaves $H$. Therefore there are at most $n$ admission rounds. Hence the epoch must terminate after finitely many rounds.

Suppose first that no failure occurs. Then eventually $P=\varnothing$, 
so every agent belongs to $H$. By \eqref{eq:engine-own}, every agent receives value at least her threshold, and by \eqref{eq:engine-efl}, the allocation is EFL among all agents. Some goods may still remain in the free pool, so the allocation may be partial with respect to the goods; \Cref{lem:completion} will allocate those remaining goods later.

Now suppose that the epoch terminates with failure. Then the chosen divider $d\in P$ receives a residual partition $\mathbf{C}=(C_1,\ldots,C_{|P|})$ of the free goods with $L=\min_jv_d(C_j)<t_d.$ Because $d$ is still unassigned immediately before the failure, the rejection invariant applies to $d$. Hence $v_d(A_j)<t_d \qquad \text{for every }j\in H.$ Thus the failure output has exactly the claimed properties: every already assigned bundle is worth strictly less than $t_d$ to the divider, and the minimum value in her residual partition is also strictly less than $t_d$.

Finally, consider the running time apart from calls to $\mathsf{Divide}$. With additive valuations, bundle values can be maintained using polynomially many arithmetic operations. Trimming candidates, testing whether an assigned agent EFL-envies a candidate, executing a repair, computing a maximum matching, and performing alternating reachability are all polynomial in $n$ and $m$. Together with the bounds of $O(nm^4)$ repair rounds and at most $n$ admission rounds, this gives a polynomial amount of non-partition work independent of the numerical magnitudes of the valuations.
\end{proof}

\subsection{Complete the allocation without losing value}

At the end of a successful epoch, every agent is assigned and meets her threshold, but some goods may remain in the free pool. The next lemma shows that all such goods can be allocated without decreasing any agent's value or destroying EFL.
We defer its proof to \cref{sec:fixed thresholds appendix}.

\begin{restatable}{lemma}{lemcomplete}
\label{lem:completion}
[Value-preserving EFL completion~\cite{F25}]
Let $\mathbf{A}=\langle A_1,\ldots,A_n\rangle$ be a partial EFL allocation with $H=N$, so every agent is assigned but some goods may remain free. Then $\mathbf{A}$ can be completed in polynomially many arithmetic operations without decreasing any agent's own value. 
\end{restatable}

\section{The two APS--MXS algorithms}\label{sec:aps-algorithms}

   In this section, we provide \cref{alg:joint}, which gets $r\in (0,1]$ as an input
    and returns a complete EFL allocation satisfying $v_i(A_i)\ge \beta_n \tau_i, \qquad v_i(A_i)\ge r \xi_i \qquad(i\in[n]).$ Using this, we prove \cref{thm:main-poly} and \cref{thm:main-exact}.

Before starting the epochs, compute the fixed floor
\begin{equation}\label{eq:apsmx-fixed-floor}
t_i^*=\widehat t(v_i,n)
\qquad(i\in N)
\end{equation}
using the event sweep in \Cref{lem:apsx-optimize}.  By
\Cref{lem:apsx-optimize,thm:apsx-split}, these floors satisfy $\beta_n\tau_i\le t_i^*\le\MMS_n(v_i;M)$ and have $\LGLB$ as a residual splitter.  They are computed once and never
changed.  This preprocessing uses $O(n(m+1)^2)$ arithmetic operations and
comparisons and does not compute APS.

For each agent $i$, we maintain a certificate value $b_i$ together with an explicit MXS upper witness whose distinguished bundle has value exactly $b_i$. 
Initially, both are obtained by applying $\LGLB$ to the entire set of goods from that agent's perspective.
The key property of the certificate is $b_i\ge \max\{\xi_i,t_i^*\}\ge\max\{\xi_i,\beta_n\tau_i\}.$ The algorithm then proceeds through a sequence of epochs. At the beginning of each epoch, the current certificates $b_i$ are converted into the thresholds $t_i=\max\{t_i^*,rb_i\}$ given to the EFL engine. The EFL engine starts again from the empty allocation and uses $\LGLB$ whenever a divider is asked to partition the free goods.

If the epoch succeeds, the remaining free goods are allocated using the completion procedure (Lemma~\ref{lem:completion}), and the algorithm terminates.

If an epoch fails on divider $i$, the bundles currently assigned to other agents, together with the residual $\LGLB$ partition produced by the failed division, form an $n$-partition of the original set of goods. Applying \Cref{thm:joint-certificate} to this partition yields a new certificate $b_i'$ satisfying $\max\{\xi_i,t_i^*\}\le b_i'<t_i=rb_i,$ where the equality $t_i=rb_i$ follows because failure is impossible at the fixed floor $t_i^*$, as proved below. We replace $b_i$ by $b_i'$ and restart the EFL engine. Hence every failed epoch decreases one certificate while preserving both guarantees encoded by it.

We obtain two algorithms by choosing the operative thresholds differently. In the polynomial-time algorithm, we set $t_i=\max\{t_i^*,(1-\eta)b_i\}$ for a suitable $\eta>0$. The APS floor is not multiplied by $1-\eta$. A failure then gives $b_i'<(1-\eta)b_i,$ so each update contracts the corresponding certificate by a multiplicative factor of at most $1-\eta$. Together with the lower bound $b_i\ge\beta_n\tau_i$, this yields only polynomially many updates.

For integer valuations, we instead set $t_i=b_i.$ A failed epoch then gives $b_i'<b_i$. Since every certificate is an integer, it decreases by at least one. This yields a pseudopolynomial bound on the number of restarts. Because a successful epoch allocates value at least $t_i=b_i\ge\xi_i$ to every agent, this second algorithm guarantees exact MXS.

\begin{algorithm}[t]
\caption{$\mathsf{JointAllocate}(r)$: fixed APS floors and $\LGLB$ certificates with retention $r$}\label{alg:joint}
\begin{algorithmic}[1]
\Require A retention parameter $r\in(0,1]$
\If{$n=1$}
    \State Allocate $M$ to the unique agent and return the one-bundle witness
\EndIf
\For{each agent $i\in N$}
    \State $t_i^*\gets\widehat t(v_i,n)$ by \Cref{lem:apsx-optimize}
    \State $\mathbf{W}^i=(W^i_1,\ldots,W^i_n)\gets\LGLB(v_i,M,n)$
    \State $b_i\gets\min_{j\in[n]}v_i(W^i_j)$
    \State Save $\mathbf{W}^i$ with a distinguished minimum-value bundle
\EndFor
\Loop
    \State Set $t_i\gets\max\{t_i^*,r b_i\}$ for every $i\in N$ 
    \State Run $\Epoch(t,\LGLB)$ from the empty allocation
    \If{the epoch returns $\mathsf{Success}(\mathbf{A},F)$}
        \State Apply Lemma~\ref{lem:completion} to $(\mathbf{A},F)$
        \State \Return the resulting complete allocation and the saved witnesses
    \Else
        \State Let $\mathsf{Failure}(d,\mathbf{A},\mathbf{C},L)$ be the returned failure
        \State Let $H$ be the assigned-agent set at the time of failure
        \State $b_d\gets\max\left(\{L\}\cup\{v_d(A_j):j\in H\}\right)$
        \State Save partition$\bigl(\big((A_j)_{j\in H},C_1,\ldots,C_{n-|H|}\big),\mathbf{C}\bigr)$ with a bundle of value $b_d$ distinguished for agent $d$ 
        \State Release all allocated goods and start a new epoch
    \EndIf
\EndLoop
\end{algorithmic}
\end{algorithm}

\subsection{Certificates, operative thresholds, and restarts}

\begin{lemma}[A fixed APS floor is never crossed]\label[lemma]{lem:apsmx-floor}
Fix an agent $i$, let $D_1,\ldots,D_k$ be arbitrary disjoint bundles,
where $0\le k<n$, and let $L$ be the minimum value of $\LGLB\left(v_i,M\setminus\bigcup_{j=1}^kD_j,n-k\right).$ Then
\begin{equation}\label{eq:apsmx-floor-certificate}
\max\bigl(\{L\}\cup\{v_i(D_j):j\in[k]\}\bigr)\ge t_i^*.
\end{equation}
Moreover, an epoch cannot fail on agent $i$ when $t_i=t_i^*$.
Once an agent's operative threshold equals $t_i^*$, it remains at that
floor and that agent causes no further restarts.
\end{lemma}
\begin{proof}
The bound \eqref{eq:apsmx-floor-certificate} is the fixed-floor part of \Cref{thm:joint-certificate}.

If $t_i=t_i^*>0$ and $i$ is an unassigned divider, the rejection invariant
makes every assigned bundle worth strictly less than $t_i^*$.  The same
residual-splitting theorem gives $L\ge t_i^*=t_i$, excluding failure.
If $t_i=t_i^*=0$, agent $i$ is assigned at the beginning of the epoch and
is never a divider.  An agent's certificate changes only upon its own
failure; hence, once its threshold equals the fixed floor, it stays there
in every later epoch.
\end{proof}

The following lemma establishes that the certificates maintained by the algorithm remain valid, characterizes the progress made after every failed epoch, and gives the fairness guarantees obtained whenever an epoch succeeds.

\begin{lemma}[Conditional output guarantees]\label[lemma]{lem:joint-output}
Throughout \Cref{alg:joint}, every certificate $b_i$ has a valid saved MXS upper witness of value exactly $b_i$ and satisfies
\begin{equation}\label{eq:certificate-invariant}
b_i\ge\max\{\xi_i,t_i^*\}\ge\max\{\xi_i,\beta_n\tau_i\}.
\end{equation}
Moreover, whenever $b_d$ is updated to $b_d'$, we have:
\begin{equation}\label{eq:update-bounds}
b_d'
<
r b_d.
\end{equation}
Whenever the algorithm terminates, its output is complete and EFL and satisfies
\begin{equation}\label{eq:output-retention}
v_i(A_i)
\ge
\max\{t_i^*,r b_i\}
\ge
\max\{\beta_n\tau_i,r\xi_i\}
\qquad(i\in[n]).
\end{equation}
\end{lemma}

\begin{proof}
At initialization, \Cref{cor:initial-range} gives $b_i\ge\max\{\xi_i,t_i^*\}\ge\max\{\xi_i,\beta_n\tau_i\},$ and the initial $\LGLB$ partition is an MXS upper witness of value exactly $b_i$.

At the beginning of each epoch, the operative thresholds are
\begin{equation}\label{eq:operative-threshold}
t_i=\max\{t_i^*,r b_i\}.
\end{equation}

The only case where some $b_d$ gets updated is when
an epoch fails on divider $d$, with assigned-agent set $H$, residual $\LGLB$ partition $\mathbf{C}$, and residual minimum $L$. The algorithm sets
\begin{equation}\label{eq:certificate-update}
b_d'
=
\max\left(
\{L\}
\cup
\{v_d(A_j):j\in H\}
\right).
\end{equation}
The assigned bundles $(A_j)_{j\in H}$ are pairwise disjoint, and their complement is exactly the free pool partitioned by $\mathbf{C}$ into $n-|H|$ bundles. Hence \Cref{thm:joint-certificate} applies to the assigned bundles and $\mathbf{C}$. It gives a new MXS upper witness of value exactly $b_d'$ and $b_d'\ge\max\{\xi_d,t_d^*\}\ge\max\{\xi_d,\beta_n\tau_d\}.$ All other certificates and witnesses remain unchanged, so \eqref{eq:certificate-invariant} is preserved.

Moreover, failure gives $L<t_d$, while the rejection invariant gives $v_d(A_j)<t_d \qquad(j\in H),$ because $d$ is still unassigned. Therefore every term in the maximum defining $b_d'$ is below $t_d$, so $b_d'<t_d$.
By \Cref{lem:apsmx-floor}, failure is impossible when $t_d=t_d^*$.
Thus a failure necessarily satisfies $t_d>t_d^*$ and hence
$t_d=r b_d$.  Consequently $b_d'<r b_d$, which gives
\cref{eq:update-bounds}.

Finally, if an epoch succeeds, Proposition~\ref{prop:epoch} gives an EFL partial allocation satisfying $v_i(A_i)\ge t_i=\max\{t_i^*,r b_i\} \qquad(i\in N).$ By Lemma~\ref{lem:completion}, the remaining goods can be allocated without decreasing any agent's value or destroying EFL. The resulting allocation is therefore complete and EFL, and combining the inequality above with \eqref{eq:certificate-invariant} gives \eqref{eq:output-retention}.
\end{proof}

The lemma proves correctness conditional on termination. It remains to bound the number of failed epochs. The next two subsections do this for two different choices of the retention parameter $r$.

\subsection{Multiplicative slack gives polynomial time}

For \Cref{thm:main-poly}, define $r=1-\epsilon.$ Thus the operative threshold is $t_i=\max\{t_i^*,(1-\epsilon)b_i\}.$ Importantly, $t_i$ itself need not be an upper bound on MXS. The quantity known to upper-bound MXS is the certificate $b_i$. Whenever the fixed APS floor does not determine $t_i$, the operative threshold is $(1-\eta)b_i$, forcing geometric progress after every failure. At the floor, failure is impossible. 

\begin{lemma}[Geometric bound on updates]\label[lemma]{lem:geometric}
Each agent with positive APS has fewer than $1/(2\epsilon)$ certificate updates. An agent with zero APS has no certificate updates. Consequently, the total number of epochs is at most $1+\left\lceil\frac{n}{2\epsilon}\right\rceil.$ \end{lemma}

\begin{proof}
Fix an agent $i$.

First suppose $\tau_i=0$.
By \Cref{cor:initial-range}, we get $b_i^{(0)}=0$.
The same initialization interval gives $0\le t_i^*\le b_i^{(0)}=0$. Hence the operative threshold of agent $i$ is zero in every epoch. Such an agent belongs to the initial assigned set $H=\{j:t_j=0\}$ and therefore is never an unassigned divider. Since certificates are updated only for failed dividers, $b_i$ is never updated.

Now suppose $\tau_i>0$.
Let $b_i^{(k)}$ denote the certificate after agent $i$ has been updated exactly $k$ times. Every update satisfies, by \eqref{eq:update-bounds},
$b_i^{(k)}<r b_i^{(k-1)}$.
Iterating this inequality gives
$b_i^{(k)}< r^k b_i^{(0)}$.

At the same time, the certificate invariant gives
$b_i^{(k)}\ge\beta_n\tau_i$,
while \Cref{cor:initial-range} gives
$b_i^{(0)}\le\tau_i$.
Therefore, for every $k\ge1$, $\beta_n\tau_i \le b_i^{(k)} < r^k b_i^{(0)} \le r^k\tau_i.$ Since $\tau_i>0$, we may cancel it and obtain $\beta_n<r^k=(1-\epsilon)^k.$ Taking logarithms gives $k < \frac{\log(1/\beta_n)} {-\log(1-\epsilon)}.$ We now bound this ratio. Since $\beta_n=\frac{2n}{3n-1},$ we have $\frac1{\beta_n} = \frac{3n-1}{2n} < \frac32.$ Hence $\log(1/\beta_n)<\log(3/2)<\frac12.$ Also, $-\log(1-\epsilon)\ge\epsilon.$ Consequently, $k < \frac{1/2}{\epsilon} = \frac1{2\epsilon}.$ Thus every positive-APS agent is updated fewer than $1/(2\epsilon)$ times, while every zero-APS agent is updated zero times. Summing over all agents, the total number of failed epochs is less than $\frac{n}{2\epsilon}.$ There is at most one additional final successful epoch, which yields the stated bound.
\end{proof}

\begin{proof}[Proof of \Cref{thm:main-poly}]
The fixed APS floors are computed once using $O(n(m+1)^2)$ arithmetic operations and comparisons by \Cref{lem:apsx-optimize}.
By Proposition~\ref{prop:epoch}, every epoch contains only polynomially many repair and admission rounds. Each call to $\LGLB(v,F,q)$ is also polynomial-time: the goods can be sorted in nonincreasing order of $v$-value using
$
O(m\log(m+1))
$
comparisons, and then each good can be assigned to a currently minimum-value bundle using a direct scan of at most $n$ bundle values. This requires at most
$O(nm)$
additional arithmetic operations and comparisons.

By \Cref{lem:geometric}, the number of epochs is polynomial in $n$ and $1/\epsilon$, independently of the numerical magnitudes of the valuations. Hence a successful epoch is eventually reached and the algorithm terminates.
Every certificate $b_i$ is an original subset value. The only additional numerical quantities are the fixed floors $t_i^*$ and products $(1-\epsilon)b_i$, all of polynomial encoding length. Therefore, its running time is polynomial in $n,m,\ell,1/\varepsilon$.

We next prove the fairness guarantees. By \Cref{lem:joint-output}, $v_i(A_i) \ge (1-\epsilon)b_i \ge (1-\epsilon)\xi_i.$ Thus the final allocation is $(1-\varepsilon)$-MXS.

For APS, \Cref{lem:joint-output,lem:apsx-optimize} gives $v_i(A_i)\ge t_i^*\ge\beta_n\tau_i =\frac{2n}{3n-1}\tau_i.$ Thus the allocation retains the entire finite-agent APS factor for every
$\varepsilon$.  The multiplicative slack affects only the MXS certificate;
it is not applied to the fixed APS floor.

The allocation is EFL by \Cref{lem:joint-output}. In addition, for each agent $i$, the algorithm retains an explicit MXS upper witness whose distinguished bundle has value
$b_i\ge\xi_i$.
This completes the proof.
\end{proof}

\subsection{No slack gives exact MXS in pseudopolynomial time}

For \Cref{thm:main-exact}, assume that all valuations are integer-valued and run \Cref{alg:joint} with
$r=1$. Then
$t_i=\max\{t_i^*,b_i\}=b_i$
in every epoch, by \eqref{eq:certificate-invariant}. 

\begin{proof}[Proof of \Cref{thm:main-exact}]
Recall that $U_i=v_i(M), \qquad U=\max_i U_i.$ For the initial $\LGLB$ partition of agent $i$, the minimum bundle value is at most the average bundle value. Hence $b_i^{(0)} \le \frac{U_i}{n}.$ Summing over the agents gives $\sum_{i=1}^n b_i^{(0)} \le \frac1n\sum_{i=1}^n U_i \le U,$ where the final inequality follows from $\sum_iU_i\le nU.$ All certificates remain nonnegative integers:
Every certificate $b_i$ is always the value of an actual subset of the original goods. This is true initially because $b_i$ is the value of a bundle of the initial $\LGLB$ partition. At an update, \eqref{eq:certificate-update} takes the maximum among $L$ and the values of assigned bundles; $L$ itself is the value of a bundle of the failed residual partition. Thus the update merely selects another subset value. 

On every failed epoch exactly one certificate is replaced by a strictly smaller integer, while all other certificates remain unchanged. Therefore $\sum_i b_i$ decreases by at least one after every failure. 
Since its initial value is at most $U$, there can be at most $U$ failed epochs and hence at most $U+1$ epochs in total.

The polynomial work inside an epoch is exactly the same as in the previous subsection; its analysis did not use integrality or a lower bound on utility increments. Consequently, the total running time is $(U+1)\poly(n,m,\ell).$ Finally, \Cref{lem:joint-output} with $r=1$ gives $v_i(A_i) \ge b_i \ge \max\{\xi_i,\beta_n\tau_i\} \qquad \text{for every }i.$ Therefore the allocation satisfies exact MXS and $\beta_n\text{-APS} = \frac{2n}{3n-1}\text{-APS}.$ It is also complete and EFL. 
The algorithm also returns, for every agent, the currently saved explicit MXS upper witness.
This completes the proof.
\end{proof}

\begin{remark}
    The fixed floors from \Cref{sec:certificate} also give a simpler algorithm when no MXS guarantee is required. In this case there is no need to maintain or update the witnesses $b_i$: one epoch at the fixed thresholds suffices.    
\end{remark}

\section{Exact RMMS and EFL for a fixed number of agents}\label{sec:rmms}
The RMMS construction uses the same rejection-and-restart architecture but a different divider certificate. It must not be obtained by renaming the $\LGLB$ threshold: an $\LGLB$ partition is not assumed to reach RMMS.

\subsection{Residual self-feasibility and attainment}

\begin{observation}[Integer bounds and existence of the maximum \cite{AFMMR26}]
\label{lem:rmms-integer}
For nonnegative integer additive valuations, $0\le\RMMS_n(v;M)\le\MMS_n(v;M)\le\left\lfloor v(M)/n\right\rfloor,$ and RMMS is an attained integer.
\end{observation}
\begin{proof}
Threshold zero is feasible: the no-deletion case has any nonnegative partition, and no bundle has value strictly below zero, so nonempty deletion cases are vacuous. Taking $k=0$ in residual self-feasibility gives the upper bound by MMS, and integrality and total value give the remaining bound.

If $t$ is feasible and $t\in(z,z+1)$ for an integer $z\ge0$, every bundle has integral value. Thus $v(D)<t$ is equivalent to $v(D)<z+1$, and $v(C)\ge t$ is equivalent to $v(C)\ge z+1$. The same residual conditions certify $z+1$. Every feasible noninteger threshold can therefore be rounded upward to a feasible integer. Only finitely many integers up to $\lfloor v(M)/n\rfloor$ can be feasible, and zero is one of them. Their maximum is the required attained RMMS.
\end{proof}

\begin{observation}
[Residual inheritance]\label{lem:inheritance}
Let $\rho=\RMMS_n(v;M)$. If $D_1,\ldots,D_k$ are disjoint, $k<n$, and $v(D_j)<\rho$ for all $j$, then $\RMMS_{n-k}\left(v;M\setminus\bigcup_jD_j\right)\ge\rho.$ \end{observation}
\begin{proof}
Take any additional $\ell<n-k$ disjoint bundles of value below $\rho$ from the residual instance. Together with the original $k$ bundles, these are $k+\ell<n$ admissible deletions from $M$. Residual self-feasibility of $\rho$ in the original instance supplies a partition of the remaining goods into $n-k-\ell$ bundles, each worth at least $\rho$. 
\end{proof}

\subsection{A failed division certifies a smaller upper threshold}
A \emph{residual divider} on $(v,F,q)$ returns a complete $q$-partition with minimum $L$ satisfying
\begin{equation}\label{eq:rmms-divider}
 L\ge\RMMS_q(v;F).
\end{equation}
It need not evaluate RMMS or maximize its returned minimum. An exact MMS partition is a sufficient implementation by Observation  \ref{lem:rmms-integer}.

\begin{lemma}[Safe RMMS threshold decrease]\label{lem:rmms-update}
Suppose $t\ge\rho=\RMMS_n(v;M)$ and $D_1,\ldots,D_k$ are disjoint, $k<n$, with $v(D_j)<t$. Let a divider satisfying \eqref{eq:rmms-divider} return minimum $L<t$ on their complement. Then
\begin{equation}\label{eq:rmms-decrease}
 \rho\le b:=\max\{L,v(D_1),\ldots,v(D_k)\}<t.
\end{equation}
\end{lemma}
\begin{proof}
Every value in the maximum is below $t$, proving the strict decrease. If $\rho>b$, every deleted bundle has value strictly below $\rho$, so \cref{lem:inheritance} gives residual RMMS at least $\rho$. The divider must then return $L\ge\rho>b$, contradicting $L\le b$. Thus $b\ge\rho$. No value of RMMS was evaluated.
\end{proof}

Initialize each upper threshold to
\begin{equation}\label{eq:rmms-start}
 t_i=\lfloor U_i/n\rfloor\ge\rho_i.
\end{equation}
In a failed division use the actually assigned bundles as $D_1,\ldots,D_k$, lower $t_d$ by \eqref{eq:rmms-decrease}, and restart. Rejection guarantees the strict inequalities required by the lemma. The partition routine, rather than a comparison with an unknown RMMS value, determines when to update.

\subsection{The elementary EFX variant of an epoch}
For clarity and fidelity to the residual-share construction, we use an inclusion-minimal EFX repair in this section rather than the EFL-specific repair of \cref{alg:repair}. It is sufficient for the fixed-agent pseudopolynomial theorem and gives a simple integer progress bound.

Start an epoch with empty bundles, $H=\{i:t_i=0\}$, and $P=N\setminus H$. Maintain \eqref{eq:engine-own} and \eqref{eq:engine-reject}, together with EFX among assigned agents. Use the RMMS-feasible divider and form minimal desired candidates exactly as in \cref{sec:engine}. The proper-subset inequality \eqref{eq:minimal-candidate} therefore holds.

If some proper subset of a candidate strictly improves an assigned agent, choose such a subset and trim it until no proper subset strictly improves \emph{any} assigned agent, using the assigned utilities frozen at the start of this call. The predicate ``some assigned agent strictly prefers this set to its current bundle'' is upward closed. One deletion scan therefore computes an inclusion-minimal improving set $S$. It remains a proper subset of its candidate. Give $S$ to any assigned agent who strictly prefers it, and release that agent's old bundle.

To detect this branch it is enough to test all single-good deletions $B-g$ of each candidate: every proper subset is contained in such a deletion. If none improves any assigned agent, no smaller proper subset can do so. Empty sets never strictly improve an assigned agent.

If the repair branch does not apply, allocate a nonempty rejection-preserving matching from \cref{lem:matching}.

\begin{lemma}[The RMMS epoch invariants]\label{lem:rmms-epoch}
Every successful round preserves threshold satisfaction, rejection, and assigned-agent EFX. A repair strictly improves one assigned agent and changes no other assigned utility. An admission assigns at least one new agent. Upon success the epoch returns a partial EFX allocation meeting all thresholds.
\end{lemma}
\begin{proof}
In a repair, $S$ is a proper subset of a candidate, so every unassigned agent still rejects it by \eqref{eq:minimal-candidate}. Its recipient strictly improves and remains above its threshold. For every assigned observer $i$ and every $g\in S$, inclusion-minimality gives
$v_i(S-g)\le v_i(A_i)$ at the old allocation. This is exactly the incoming EFX requirement for the replacement bundle. The recipient's outgoing comparisons persist because it improves, and other allocated bundles are unchanged or released.

In an admission, old assigned agents are EFX toward every candidate: otherwise a one-good deletion would trigger the repair branch. A new recipient meets its threshold and does not envy old assigned bundles by rejection. Every proper subset of another new candidate lies below its threshold, so its EFX comparisons with new recipients also hold. The matching preserves rejection by remaining unassigned agents. The invariants are initially true, including for zero-threshold agents. Induction proves the claim, and at success all agents are assigned.
\end{proof}

\begin{algorithm}[t]
\caption{Descending upper thresholds for RMMS and EFL}\label{alg:rmms}
\begin{algorithmic}[1]
\Require $n,m,$ and valuation function $v_i$ for every agent $i$
\If{$n=1$}\State Give the agent all goods and return.\EndIf
\State $t_i\gets\lfloor v_i(M)/n\rfloor$ for every $i$.
\Loop
 \State Reset $A_i\gets\varnothing$, $F\gets M$, and $H\gets\{i:t_i=0\}$.
 \State $\mathit{restart}\gets\mathsf{false}$.
 \While{$P:=N\setminus H$ is nonempty}
  \State Choose $d\in P$; divide $F$ into $|P|$ bundles satisfying \eqref{eq:rmms-divider}.
  \State Let $L$ be their minimum under $v_d$.
  \If{$L<t_d$}
   \State $t_d\gets\max\bigl(\{L\}\cup\{v_d(A_j):j\in H\}\bigr)$.
   \State $\mathit{restart}\gets\mathsf{true}$; break the inner loop.
  \EndIf
  \State Trim the bundles into inclusion-minimal desired candidates.
  \If{a proper subset of a candidate improves an assigned agent's value}
   \State Trim to a minimal improving set $S$; give it to an improving assigned agent.
   \State Release its old bundle and discard the temporary candidates.
  \Else
   \State Allocate a nonempty rejection-preserving matching; admit its recipients.
  \EndIf
 \EndWhile
 \If{not $\mathit{restart}$}
  \State \Return $\Complete(A,F)$.
 \EndIf
\EndLoop
\end{algorithmic}
\end{algorithm}

\begin{lemma}[A pseudopolynomial number of rounds]\label{lem:rmms-rounds}
The algorithm has at most $U$ threshold decreases. Each epoch has at most $W$ successful rounds and at most one failed division, where $W=\sum_iU_i\le nU$.
\end{lemma}
\begin{proof}
By \cref{lem:rmms-update}, every update preserves $t_i\ge\rho_i$ and strictly reduces one nonnegative integer threshold. Initially $\sum_i t_i=\sum_i\lfloor U_i/n\rfloor\le W/n\le U.$ Thus there are at most $U$ updates.

Within a fixed epoch use the potential
$\Psi=\sum_{i\in H}v_i(A_i)$. It starts at zero, never decreases, and is at most $W$. A repair increases it by a positive integer. Every admission gives at least one new agent a bundle meeting a positive integer threshold, so it also increases $\Psi$ by at least one. Therefore at most $W$ successful rounds occur. A failed division immediately ends the epoch. No comparison of $\Psi$ across restarts is used.
\end{proof}

\begin{theorem}[Residual-divider reduction]\label{thm:rmms-reduction}
Suppose a divider satisfying \eqref{eq:rmms-divider} runs in time at most $T(n,m,U)$ for every residual instance under consideration. Then \cref{alg:rmms} returns exact RMMS and EFL with running time
\begin{equation}\label{eq:rmms-oracle-time}
 O\left((U+1)(W+1)\,[T(n,m,U)+nm^2+n^3]
       +(n+m+1)^6\right).
\end{equation}
It uses no RMMS-value oracle.
\end{theorem}
\begin{proof}
Initialization and every update preserve $t_i\ge\rho_i$. The epoch invariants and round bound preclude an infinite final epoch. If an unassigned agent remains, the next division either fails, causing another bounded threshold decrease, or produces a valid successful round. The process therefore eventually assigns every agent and returns a partial EFX allocation above all $t_i$. Since EFX implies EFL, Lemma~\ref{lem:completion} shows that it can be completed to EFL without losing any achieved utility.

There are at most $(U+1)(W+1)$ divider calls. Trimming, detecting a proper improving subset, and minimal-improvement trimming cost $O(nm^2)$ per round with direct additive evaluation; matching costs $O(n^3)$. Completion is bounded by the final term in \eqref{eq:rmms-oracle-time}. All ordinary numeric data are integer subset values or thresholds at most $U$, apart from polynomial-size indices. This completes the proof.
\end{proof}

\subsection{Implementing the residual divider}

For a single nonnegative additive valuation, computing an
MMS-optimal partition means dividing the goods into $q$ bundles
so as to maximize the value of the least valuable bundle.
For a fixed number of bundles $q$, an exact
pseudopolynomial-time algorithm based on dynamic programming
is known~\cite{ChiarelliEtAl2023}. We briefly describe this
dynamic program and analyze its running time.

\begin{lemma}[Exact MMS partitioning]\label{lem:dp}
Let $v$ be a nonnegative integer additive valuation on a set $F$ of
$f$ goods, and write $S=v(F)$. An MMS-optimal $q$-partition of $F$
can be computed in $O\bigl(q^2(f+1)(S+1)^{q-1}\bigr)$ elementary arithmetic and indexing operations, with polynomial
operand lengths.
\end{lemma}

\begin{proof}
For $q=1$, the unique partition consists of $F$ itself,
so the claim is immediate. Assume henceforth that $q\ge 2$.

Let $F=\{g_1,\ldots,g_f\}$, and write
$w_j=v(g_j)$ and $S=v(F)$.
We construct a dynamic program that processes the goods
one by one. After processing the first $j$ goods, a state
$(x_1,\ldots,x_{q-1})$ records the values of the first
$q-1$ bundles in some partition of these goods.
We do not need to store the value of the last bundle,
since it is uniquely determined by $x_q=\sum_{h=1}^{j}w_h-\sum_{a=1}^{q-1}x_a.$ Initially, no goods have been processed, so the only
state is $(0,\ldots,0)$. When processing good $g_j$,
we consider all $q$ possible bundles to which it can
be assigned. Assigning it to bundle $a<q$ increases
coordinate $x_a$ by $w_j$, while assigning it to
bundle $q$ leaves the stored vector unchanged.

We first show that the dynamic program is exact.
By induction on $j$, the states after processing the
first $j$ goods are precisely the load vectors
attainable by partitioning those goods into $q$
bundles. Indeed, every partition of the first $j$
goods is obtained by assigning $g_j$ to one of the
$q$ bundles of a partition of the first $j-1$ goods.
Conversely, every transition corresponds to a valid
assignment of $g_j$.

Moreover, two partial partitions that yield the same
state have identical loads in all $q$ bundles, since
the last load is determined by the processed total.
Therefore, merging such states does not remove any
possible final load vector.

After all goods have been processed, we select a
reachable state maximizing
\[
  \min\left\{
    x_1,\ldots,x_{q-1},
    S-\sum_{a=1}^{q-1}x_a
  \right\}.
\]
Since every $q$-partition of $F$ is represented by
a reachable state, this maximum equals
$\MMS_q(v;F)$. By storing a predecessor and the
chosen bundle for each state, we can recover an
optimal partition by backtracking.

It remains to bound the running time.
Every coordinate is an integer between $0$ and $S$.
Thus, there are at most $(S+1)^{q-1}$ states in
each of the $f+1$ layers.

Each state has $q$ possible transitions, and
processing a transition requires at most $O(q)$
arithmetic and indexing operations. The total
running time is therefore $O\bigl(q^2(f+1)(S+1)^{q-1}\bigr).$ Finally, all stored loads are integers bounded
by $S$, so their binary encoding lengths are
$O(\log(S+2))$. Hence all arithmetic operations
have polynomial operand lengths.
\end{proof}

\begin{proof}[Proof of \Cref{thm:main-rmms}]
For each residual instance $(v,F,q)$, the partition
$\C=(C_1,\ldots,C_q)$ returned by \cref{lem:dp} satisfies $\min_{j\in[q]}v(C_j)=\MMS_q(v;F)\ge\RMMS_q(v;F).$ Thus the divider requirement~\eqref{eq:rmms-divider} holds. Since
$q\le n$, $f\le m$, and $S\le U$, the divider runs in $T(n,m,U)=O\bigl(n^2(m+1)(U+1)^{n-1}\bigr)$ arithmetic and indexing operations. Substituting this bound into
\eqref{eq:rmms-oracle-time} in \cref{thm:rmms-reduction}, using
$W\le nU$, and accounting for polynomial operand lengths gives
bit complexity $\poly(n,m,U,\ell)(U+1)^{n-1}.$ For fixed $n$, this is pseudopolynomial. The algorithm does not
compute RMMS values; it only uses exact MMS partitions of residual
instances.
\end{proof}

\begin{remark}[Sufficient partition primitive]
The reduction requires only~\eqref{eq:rmms-divider}, not MMS
optimality. Any divider satisfying this condition can replace the
exact MMS routine in \cref{thm:rmms-reduction}; the standard dynamic
program above suffices for the fixed-agent pseudopolynomial bound.
\end{remark}

\section{Hardness of selecting a certificate between MXS and RMMS}
\label{app:mxsh:hardness}
\paragraph{Motivation.}
It is natural to ask whether the interval between MXS and RMMS contains a share that is easier to compute than either endpoint.  Such a share would be algorithmically useful: a value $s_i$ satisfying $\operatorname{MXS}_n(v_i;M) \le s_i \le \operatorname{RMMS}_n(v_i;M)$ would already dominate the MXS requirement while remaining within the residual-feasibility range captured by RMMS.  If such a value could be computed in polynomial time together with an appropriate residual splitter, it could be used as a fixed numerical threshold in the residual allocation framework and combined with EFL and other residual-valid guarantees, such as the APS floors developed in this paper.  This motivates searching for a polynomially computable intermediate certificate rather than attempting to compute MXS or RMMS exactly.

Computing the exact minimum EFX share is NP-hard even for two agents~\cite[Theorem~5]{CGRSV22}.
That fact alone does not imply hardness of computing an intermediate share:
for two agents, the minimum value of a one-valuation greedy EFX partition
lies between MXS and RMMS. This section establishes hardness for the
intermediate-share problem itself. The reduction uses three agents and
forces MXS and RMMS to coincide, with their common value distinguishing
the two outcomes of a \textsc{Partition} instance.

\begin{remark}
    Akrami et al.~\cite{AFMMR26} define the \emph{strong EEFX share} as the
    value threshold beyond which EEFX-feasibility becomes monotone.  It is
    therefore natural to ask whether one can replace this share by a
    polynomially computable certificate share that is at least as large, and
    then feed this certificate into the allocation framework of this paper.
    However, since this value lies between MXS and RMMS this particular route cannot work in
    polynomial time either, unless $\mathsf{P}=\mathsf{NP}$.  The obstruction is not
    to EEFX itself; it is to simultaneously requiring a numerical threshold
    to dominate the strong EEFX share and to be self-feasible (hence usable as
    a residual allocation threshold).
\end{remark}

\mxsrmmsenhard*

\subsection{Definitions and the share comparison}
\label{subsec:mxsh:definitions}
    
    A bundle $T$ is EEFX-feasible if its complement has an $(n-1)$-partition
    whose every bundle has EFX score at most $v(T)$. The minimum value of such
    a $T$ is $\MXS_n(v;M)$. A threshold $t$ is residual-valid if, for every
    $k<n$ and every collection of disjoint bundles $D_1,\ldots,D_k$ with
    $v(D_j)<t$, the complement of their union has an $(n-k)$-partition of
    minimum value at least $t$. RMMS is the largest residual-valid threshold.
    The following known comparison~\cite{AFMMR26} is recalled with an additive
    proof to make the reduction self-contained.
    
\begin{lemma}
\label{lem:mxsh:comparison}
For nonnegative additive valuations, $\MXS_n(v;M)\leq\RMMS_n(v;M)\leq\MMS_n(v;M).$ \end{lemma}
\begin{proof}
The second inequality is the no-deletion case of residual validity.
For the first, let $t=\MXS_n(v;M)$. The case $t=0$ is immediate.
The greedy rule $\LGLB$, which assigns goods in nonincreasing value order
to a currently least-valued bundle, has the following property: every
bundle's EFX score is at most the partition's minimum value. Indeed, the
last good inserted into a bundle is one of its minimum-valued goods, and
its load immediately before that insertion is no larger than every
other final load.

Suppose $k<n$ deleted bundles each have value below $t$, but the greedy
$(n-k)$-partition of the remaining goods has minimum value $\ell<t$.
The combined $n$-partition has every EFX score at most $b=\max\bigl(\{\ell\}\cup\{v(D_j):j\in[k]\}\bigr)<t.$ One of its bundles has value exactly $b$, so it is EEFX-feasible. This
contradicts the definition of $t$. This argument includes $k=0$.
Thus $t$ is residual-valid.
\end{proof}

\subsection{The three-agent construction}
\label{subsec:mxsh:construction}

Start with an instance of the NP-complete problem
\textsc{Partition}~\cite{GJ78}: positive integers
$a_1,\ldots,a_r$ with $\sum_{j=1}^{r}a_j=2B.$ The question is whether some subset sums to $B$. Restricting to even
total does not affect NP-hardness, since odd-total instances are
immediate no-instances. Define
\begin{equation}
 C=2B+1,\qquad L=4(rC+B).
 \label{eq:mxsh:parameters}
\end{equation}

    Create $2r+3$ goods, with three identical agents with the following valuation:
\begin{align}
 v(e_j)&=4(C+a_j) &&(j\in[r]),\notag\\
 v(d_j)&=4C &&(j\in[r+1]),\label{eq:mxsh:values}\\
 v(z)&=L+1, \notag\\
 v(u)&=1.\notag
\end{align}
Let $E=\{e_1,\ldots,e_r\}$, $F=\{d_1,\ldots,d_{r+1}\}$, and
$R=E\cup F$. The set $R$ is the core, $z$ is the marker good, and $u$
is the unit good. Useful identities are
\begin{equation}
 v(R)=2L+4C,\qquad
 v(E)=L+4B,\qquad
 v(F)=L+4B+4.
 \label{eq:mxsh:totals}
\end{equation}

\begin{lemma}[Core subset values]
\label{lem:mxsh:core}
If $S\subseteq R$ and
$v(S)\leq L+1$, then $v(S)\leq L$, $|S|\leq r$, and $R\setminus S$
contains a good of value $4C$. Consequently,
\begin{equation}
 \vefx_v(R\setminus S)=2L-v(S).
 \label{eq:mxsh:core-score}
\end{equation}
Moreover, there is a core subset of value exactly $L$ if and only if the
source \textsc{Partition} instance is a yes-instance.
\end{lemma}
\begin{proof}
$v(S)$ is divisible by four. Since $L$ is divisible by four, $v(S)\leq L+1$ implies $v(S)\leq L$.
Every core good is worth at least $4C$, and
$L<4(r+1)C$, so $|S|\leq r$. There are $r+1$ goods of value $4C$,
at least one of which remains outside $S$. All other core goods are
larger, proving \eqref{eq:mxsh:core-score}.

$v(S)=L$ is equivalent to $|S|C+\sum_{j:e_j\in S}a_j=rC+B.$ Since $0\leq\sum_{j:e_j\in S}a_j\leq2B$ and $C>B$, this forces $|S|=r$
and $\sum_{j:e_j\in S}a_j=B$. Conversely, any subset of the $a_j$ summing
to $B$ gives a core subset of value $L$ by padding its element goods
with value-$4C$ goods until it contains $r$ goods.
\end{proof}

\begin{lemma}[The maximin share is constant across the reduction]
\label{lem:mxsh:mms}
Every instance \eqref{eq:mxsh:values} satisfies $\MMS_3(v;M)=L+2.$ The partition $(\{z,u\},E,F)$ is EFX and an MMS partition.
\end{lemma}
\begin{proof}
By \eqref{eq:mxsh:totals}, this partition's values are
$L+2$, $L+4B$, and $L+4B+4$, whose minimum is $L+2$ because $B\geq1$.
Its scores satisfy $\vefx_v(\{z,u\})=L+1,\quad \vefx_v(E)=L+4B-4(C+\min_j a_j)<L,\quad \vefx_v(F)=4rC<L.$ It is therefore EFX, and it gives the MMS lower bound.

For the upper bound, consider any three-partition. If the bundle containing
$z$ contains no core good, its value is at most $L+2$.
Otherwise at most $2r$ core goods remain for the other two bundles.
A bundle containing at most $r-1$ core goods has value at most $4(r-1)C+8B+1=L-4B-3<L+2,$ even if it receives $u$ and all the source-weight contribution.
Hence a partition with minimum above $L+2$ would have to put exactly $r$
core goods in each of these other two bundles. Their combined value is
then at most $8rC+8B+1=2L+1$, which is smaller than $2(L+2)$.
This is impossible.
\end{proof}

\subsection{Exact MXS and RMMS values}
\label{subsec:mxsh:gap}

\begin{lemma}[The MXS gap]
\label{lem:mxsh:mxs}
The construction satisfies
\[
 \MXS_3(v;M)=
 \begin{cases}
 L+1,&\text{if the source instance is a yes-instance},\\
 L+2,&\text{if the source instance is a no-instance}.
 \end{cases}
\]
\end{lemma}
\begin{proof}
First, no bundle $T$ of value below $L+1$ is EEFX-feasible.
Such a bundle cannot contain $z$. In any putative complement certificate,
$z$ must be a singleton: if its bundle contained another good, deleting
that other good would leave value at least $L+1>v(T)$.
Thus the other certificate bundle must contain every remaining good.
Put $S=T\cap R$:

If $u\notin T$, the other certificate bundle is $(R\setminus S)\cup\{u\}$,
and its score is $v(R)-v(S)=2L+4C-v(S)>v(S)=v(T),$ since $v(S) \leq L$. If $u\in T$, then $v(S)+1<L+1$, so divisibility gives
$v(S) \leq L-4$. Its rival's score is
$2L-v(S) \geq L+4>v(S)+1=v(T)$ by Lemma~\ref{lem:mxsh:core}.
Both cases contradict EEFX-feasibility. Thus $\MXS_3\geq L+1$.

In a yes-instance, take a core subset $S$ with $v(S)=L$.
The bundle $S\cup\{u\}$ has value $L+1$, and its complement is certified
by the two bundles $\{z\}$ and $R\setminus S$, whose scores are $0$ and
$L$, respectively. Hence $\MXS_3=L+1$.

In a no-instance, a bundle of value $L+1$ that does not contain $z$ would have to be
$S\cup\{u\}$ with $v(S)=L$, which is impossible by
Lemma~\ref{lem:mxsh:core}. The only remaining candidate of value $L+1$
is $\{z\}$. Suppose its complement $R\cup\{u\}$ had a two-bundle
certificate. Write its bundle containing $u$ as $S\cup\{u\}$. 
Its score is $v(S)$, so $v(S)\leq L+1$ and hence $v(S)\leq L$.
The other bundle's score is $2L-v(S)$ by Lemma~\ref{lem:mxsh:core}.
For both scores to be at most $L+1$, one needs $L-1\leq v(S)\leq L.$ Divisibility forces $v(S)=L$, again a contradiction with the assumption of having a no-instance. Thus every
EEFX-feasible bundle has integer value at least $L+2$.
The partition in Lemma~\ref{lem:mxsh:mms} gives an EEFX-feasible bundle
$\{z,u\}$ of value $L+2$, proving equality.
\end{proof}

\begin{lemma}[The interval collapses]
\label{lem:mxsh:rmms}
For the construction,
\begin{equation}
 \MXS_3(v;M)=\RMMS_3(v;M)=
 \begin{cases}
 L+1,&\text{in a yes-instance},\\
 L+2,&\text{in a no-instance}.
 \end{cases}
 \label{eq:mxsh:gap}
\end{equation}
\end{lemma}
\begin{proof}
In a no-instance, Lemmas~\ref{lem:mxsh:comparison},
\ref{lem:mxsh:mms}, and~\ref{lem:mxsh:mxs} give $L+2=\MXS_3\leq\RMMS_3\leq\MMS_3=L+2.$ In a yes-instance, let $S\subseteq R$ have value $L$ and delete
$D=S\cup\{u\}$, of value $L+1$. Put $Q=R\setminus S$.
The remaining goods are $Q\cup\{z\}$, where $v(z)=L+1,\qquad v(Q)=L+4C>L+1,\qquad \vefx_v(Q)=L<L+1.$ Next, we show their two-agent MMS is exactly $L+1$. For the upper bound, in any two-partition, either $z$
is alone, giving a bundle of value $L+1$, or its bundle contains some $g\in Q$.
In the latter case, the other bundle is a subset of $Q\setminus\{g\}$
and has value at most $\vefx_v(Q)=L<L+1$.

Therefore every threshold $t>L+1$ fails residual validity: the single
bundle $D$ has value $L+1<t$, while its complement has two-agent MMS
$L+1<t$. It follows that $\RMMS_3\leq L+1$. The lower bound follows
from $\MXS_3=L+1$ and Lemma~\ref{lem:mxsh:comparison}.
\end{proof}

\paragraph{Polynomial reduction.}
The construction has $2r+3$ goods, all with positive integer values
$O(rB)$, and polynomial binary encoding length. By
\eqref{eq:mxsh:gap}, any number satisfying \eqref{eq:mxsh:sandwich}
is forced to equal $L+1$ in a yes-instance and $L+2$ in a no-instance.
Comparing the output with $L+1+1/2$ therefore decides \textsc{Partition}.
This proves Theorem~\ref{thm:mxsh:main} for three agents, including the
promise that MXS and RMMS coincide. This is an ordinary NP-hardness
reduction; it does not claim strong NP-hardness.

\subsection{More agents and nondegenerate valuations}
\label{subsec:mxsh:extensions}

\begin{lemma}[Huge-good padding preserves both endpoints]
\label{lem:mxsh:padding}
Let $v$ be a positive additive valuation on a nonempty set $M_0$, and let
$p\geq2$. Add $h$ goods, each of value $G>v(M_0)$, and increase the
number of agents from $p$ to $p+h$. Then MXS, RMMS, and MMS are unchanged.
\end{lemma}
\begin{proof}
    For MMS, appending $h$ huge singletons gives the lower bound. In any
    $(p+h)$-partition, at least $p$ bundles contain no huge good. Select $p$
    such bundles and append any other original goods to them; their minimum
    gives an original $p$-partition with at least the enlarged partition's
    minimum. This proves the upper bound.

    An original MXS witness for the instance on \(M_0\) can be extended to the enlarged instance by adding each huge good as a singleton bundle. This immediately shows that the MXS value in the enlarged instance is at most the original MXS value. For the reverse direction, consider a minimum-value MXS witness in the enlarged instance. Its value is at most the original MXS value, and hence is at most \(v(M_0)<G\). Therefore, its distinguished bundle cannot contain any huge good. Moreover, every huge good must appear as a singleton in the corresponding certificate partition, since any bundle containing a huge good has value at least \(G\), strictly larger than the value of the distinguished bundle. Removing these huge singleton bundles leaves an MXS witness for the original instance on \(M_0\). Hence, the MXS value is unchanged.
    
    Both RMMS values are at most the common MMS, which is below $G$.
    Fix a threshold $0<t<G$. Admissible deleted bundles contain no huge
    good. For $k<p$ deletions, a partition of the enlarged remainder into
    $p+h-k$ bundles meeting $t$ exists if and only if the original remainder
    can be partitioned into $p-k$ such bundles. One direction appends the huge
    singletons; the other selects $p-k$ bundles without huge goods and appends
    all remaining original goods. When $k\geq p$, the enlarged instance has
    at most $h$ remaining agents, so the huge goods alone suffice.
    Thus residual validity is equivalent in the two instances throughout
    the range containing their RMMS values.
\end{proof}

    For any fixed \(n\ge 3\), apply Lemma~\ref{lem:mxsh:padding} with \(p=3\) and \(h=n-3\). Since the exact gap in \eqref{eq:mxsh:gap} is preserved under this padding, the first part of Theorem~\ref{thm:mxsh:main} follows for every fixed \(n\ge 3\).

\begin{proposition}[Nondegenerate positive integer inputs]
\label{prop:mxsh:nondegenerate}
For every fixed $n\geq3$, the intermediate-share problem remains NP-hard
when all singleton values are positive integers and all subset values
are distinct.
\end{proposition}
\begin{proof}
Take an integer instance of the preceding reduction, including the
padding when needed. Index its $q$ goods by $g_1,\ldots,g_q$ and put
\[
 \widetilde v(g_j)=v(g_j)+\frac{2^{j-1}}{2^{q+4}},\qquad
 \Delta=\sum_{j=1}^{q}\frac{2^{j-1}}{2^{q+4}}<\frac1{16}.
\]
All strict comparisons between original integer subset values are
preserved. Equal original subset values become distinct because the
perturbation encodes the subset in binary. Thus $\widetilde v$ is
nondegenerate. Furthermore, every EEFX certificate under $\widetilde v$
is also a certificate under $v$: an originally violated comparison
has integer gap at least one and cannot be reversed by perturbations
of total size below $1/16$.

In a no-instance, Lemma~\ref{lem:mxsh:mxs} and the padding lemma give
$\MXS_n(v;M)=L+2$. Hence every new EEFX-feasible bundle has original
value at least $L+2$, and therefore $\MXS_n(\widetilde v;M)\geq L+2.$ In a yes-instance, keep the deleted bundle $D$ from
Lemma~\ref{lem:mxsh:rmms}. Its new value is at most $L+1+\Delta$.
Its original complement has $(n-1)$-agent MMS $L+1$; with padding this
follows again by the MMS part of Lemma~\ref{lem:mxsh:padding}.
Adding nonnegative perturbations of total size $\Delta$ increases the
MMS of any fixed remaining set by at most $\Delta$: for each partition,
a previously minimum-valued bundle increases by at most $\Delta$.
Thus every threshold above $L+1+\Delta$ still fails on deletion $D$, and $\RMMS_n(\widetilde v;M)\leq L+1+\Delta<L+1+\frac1{16}.$ Comparison with $L+1+1/2$ separates the two cases.
Finally, multiply every value by $Q=2^{q+4}$, obtaining positive integers
$Qv(g_j)+2^{j-1}$. All share values scale by $Q$, so compare the output
with $Q(L+1+1/2)$. The encoding length remains polynomial. This proves the
nondegenerate version and completes Theorem~\ref{thm:mxsh:main}.
\end{proof}

\subsection{Consequences and scope}
\label{subsec:mxsh:scope}

\begin{corollary}[No polynomial MXS-dominating residual certificate]
\label{cor:mxsh:residual}
Unless $P=NP$, no deterministic polynomial-time algorithm can compute,
for every positive integer additive valuation, a residual-valid floor
$t\geq\MXS_n(v;M)$. This holds already for three agents, even if the
algorithm is not required to output an implementation of a residual
splitter.
\end{corollary}
\begin{proof}
Residual validity implies $t\leq\RMMS_n(v;M)$, so such an algorithm
would solve the intermediate-share problem.
\end{proof}

Consequently, a polynomial construction of independent
MXS-dominating residual floors cannot be used as a general numerical
preprocessing step for simultaneous MXS, approximate APS, and EFL.

\section{Conclusion}\label{sec:conclusion}

We studied the simultaneous satisfaction of share-based and envy-based fairness for additive valuations, focusing on APS, MXS, RMMS, and EFL. Our results show that requiring exact EFL is compatible with substantially stronger share guarantees than was previously known. In polynomial time, we obtain exact EFL together with $(1-\varepsilon)$-MXS and $\frac{2n}{(3n-1)}$-APS. 
We also provide a pseudopolynomial algorithm simultaneously that achieves exact EFL, exact MXS, and $\frac{2n}{3n-1}$-APS, matching the strongest known general APS approximation for equal entitlements while imposing the additional MXS and EFL requirements. We also give a fixed-agent pseudopolynomial algorithm for exact RMMS together with exact EFL.
Moreover, we showed that the problem of computing a number between MXS share and RMMS share given a valuation function is
weakly NP-hard.

Several questions remain open. Most prominently, it would be interesting to determine whether exact MXS together with EFL and the $\frac{2n}{3n-1}$-APS guarantee can be computed in polynomial time, rather than pseudopolynomial time. Similarly, obtaining a polynomial-time algorithm for exact RMMS together with EFL would strengthen the algorithmic counterpart of the known existence results. More generally, our results suggest studying which share notions admit efficiently computable certificates that remain valid under the removal of allocated bundles. Developing such certificates may provide a systematic route to combining increasingly strong share guarantees with exact envy-based fairness.
Finally, it remains open whether simultaneous approximate MMS (or APS) and EFL can achieve a factor strictly above $\frac{2n}{3n-1}$.

\paragraph{AI acknowledgment.}
    The authors used ChatGPT Pro (OpenAI) as a research and writing assistant during the preparation of this paper. In particular, it was used to help explore and formalize the authors' ideas for extending the RMMS--EFL existence proof to an algorithmic, polynomial-time framework. All mathematical arguments, statements, and proofs in the paper were subsequently checked and verified, and polished by the authors, who take full responsibility for the content.

{\small
\setlength{\bibsep}{2pt}
\bibliographystyle{plainnat}
\bibliography{references}
}

\appendix
\section{Deferred proofs of \cref{app:apsx}}\label{sec:aps certificate appendix}

\characlemma*

\begin{proof}
Let $\mathcal B=\{B\subseteq M:v(B)\ge\lambda\}$ be the family of bundles whose value is at least $\lambda$. If this family is empty, both assertions are false; otherwise we argue as follows.

By the definition of APS, $\lambda\le\APS_{1/n}(v;M)$ if and only if for every normalized nonnegative price vector $p$, there exists a bundle $B\in\mathcal B$ whose price is at most $1/n$. Equivalently, $\max_{p\in\Delta(M)} \min_{B\in\mathcal B}p(B) \le\frac1n,$ where $\Delta(M)$ denotes the set of probability distributions on the goods, and $p(B)= \sum_{g\in B} p_g$.

Consider the finite zero-sum game in which the row player chooses a good $g\in M$, the column player chooses a bundle $B\in\mathcal B$, and the payoff for the row player is $1$ if $g\in B$ and $0$ otherwise. By the minimax theorem,
\[
\max_{p\in\Delta(M)}
\min_{B\in\mathcal B}p(B)
=
\min_{x\in\Delta(\mathcal B)}
\max_{g\in M}
\Pr_{B\sim x}[g\in B].
\]
Hence $\lambda\le\APS_{1/n}(v;M)$ if and only if there exists a distribution $x$ over $\mathcal B$ such that $\Pr_{B\sim x}[g\in B]\le\frac1n \qquad\text{for every }g\in M.$ Multiplying every probability in $x$ by $n$ and denoting the resulting coefficients by $\gamma_j$ gives $\sum_j\gamma_j=n$ and $\sum_{j:g\in B_j}\gamma_j\le1 \qquad\text{for every }g\in M.$ Every bundle in the support belongs to $\mathcal B$, so it has value at least $\lambda$. This proves \eqref{eq:fractional}.

We now prove the two additional properties. Consider any subfamily whose coefficients have total mass $p$. Since every bundle in the subfamily has value at least $\lambda$, $\sum_j\gamma_jv(B_j)\ge p\lambda.$ By additivity and the incidence constraints,
\[
\sum_j\gamma_jv(B_j)
=
\sum_{g\in M}
v(g)
\sum_{j:g\in B_j}\gamma_j
\le
\sum_{g\text{ used by the subfamily}}v(g).
\]
Thus the total value of the goods used by the subfamily is at least $p\lambda$.

Finally, let $G\subseteq M$ contain $h$ goods. The coefficient mass of bundles intersecting $G$ is at most $\sum_{g\in G} \sum_{j:g\in B_j}\gamma_j \le \sum_{g\in G}1 = h.$ This proves the last claim.
\end{proof}

\section{Deferred proofs of \cref{sec:engine}}
\label{sec:fixed thresholds appendix}

\lemcomplete*

\begin{proof}
The completion consists of two stages.

\paragraph{Stage 1: establish domination of every free singleton.}
Initially, an agent may value a free good more than her own bundle. While there exist an agent $i$ and a free good $g$ such that $v_i(g)>v_i(A_i),$ replace $A_i$ by the singleton $\{g\}$ and return the old bundle $A_i$ to the free pool.

The recipient strictly improves. We claim that EFL is preserved.

For every other agent $j$, $\vefl_j(\{g\})=0.$ Since valuations are nonnegative, $v_j(A_j)\ge0=\vefl_j(\{g\}),$ so no other agent EFL-envies the new singleton bundle. Every other allocated bundle remains unchanged.

For the recipient $i$, her own utility strictly increases. Therefore every EFL comparison that was satisfied before the swap remains satisfied afterward. 
Hence each singleton swap preserves EFL.

During this stage, an agent's bundle changes only when that agent herself receives a free singleton. After her first such swap, she holds a singleton. Every subsequent swap replaces it by another singleton of strictly larger value. There are at most $m$ distinct goods, so each agent can perform only polynomially many such swaps; in particular, the total number is at most $nm$.

When this stage terminates, we have
\[
v_i(g)\le v_i(A_i)
\qquad
\text{for every }i\in N
\text{ and every free good }g\in F.
\]
We call this the \emph{singleton-domination property}.

\paragraph{Stage 2: source additions.}
Construct the strict-envy graph of the current allocation. It has one vertex for each agent, and a directed edge $i\to j$ whenever $v_i(A_j)>v_i(A_i).$ If the envy graph contains a directed cycle, rotate the bundles along that cycle. Every participating agent receives a bundle she strictly prefers to her current bundle. Thus every participant strictly increases her own value, while the multiset of allocated bundles is unchanged.

A cycle rotation preserves EFL. For a participant, the value of her own bundle increases, so all of her EFL comparisons become weakly easier to satisfy. For an agent outside the cycle, her own bundle is unchanged and the multiset of rival bundles is unchanged. Hence all EFL comparisons remain valid.

The singleton-domination property is also preserved because the set of free goods does not change and no agent's own utility decreases.

Repeatedly eliminate envy cycles until the envy graph is acyclic. Every finite acyclic directed graph contains a vertex of indegree zero. Choose such an agent $s$. No other agent strictly envies $s$, so $v_i(A_s)\le v_i(A_i) \qquad \text{for every }i\ne s.$ Now choose any free good $g\in F$ and add it to $A_s$. We verify that EFL is preserved.

Fix another agent $i\ne s$. Since $s$ has indegree zero, $v_i(A_s)\le v_i(A_i).$ By singleton domination, $v_i(g)\le v_i(A_i).$ Using $g$ in the definition of the EFL score of $A_s+g$, we obtain $\vefl_i(A_s+g) \le \max\{v_i(g),v_i(A_s)\} \le v_i(A_i).$ Therefore agent $i$ does not EFL-envy the enlarged bundle $A_s+g$.

Every other rival bundle is unchanged. Agent $s$ weakly increases her own utility by receiving $g$, so her own EFL comparisons remain valid as well. Thus the source addition preserves EFL.

Remove $g$ from the free pool and repeat: eliminate envy cycles, select an indegree-zero agent, and give that agent one free good. The process terminates after at most $m$ source additions.

Between two source additions, the multiset of allocated bundles is fixed. Each cycle rotation strictly improves every participating agent. Since there are only $n$ bundles in the current multiset, each agent can move through at most $n$ distinct bundle values before the next source addition. Hence only polynomially many cycle rotations can occur between two additions.

Therefore the entire completion procedure uses polynomially many arithmetic operations, never decreases any agent's own value, preserves EFL throughout, and terminates with a complete allocation. 
\end{proof}

Simply adding free goods to sources without first establishing singleton domination need not preserve EFL. This is why the completion procedure includes the first stage.

\end{document}